%% file: main.tex
\documentclass{article}

\usepackage[utf8]{inputenc}

\title{Cross-Chain Payment Protocols with Success Guarantees}
\author{Rob van Glabbeek\\{\small Data61, CSIRO, Sydney, Australia}\\{\small UNSW, Sydney, Australia}\and Vincent Gramoli\\{\small University of Sydney, Australia}\\{\small Data61, CSIRO, Sydney, Australia} \and Pierre Tholoniat\\{\small University of Sydney, Australia}\\{\small \'{E}cole Polytechnique, France}}
\date{}

\usepackage{geometry}
\usepackage{graphicx}
\usepackage{amsfonts}
\usepackage{algorithm}
\usepackage{algpseudocode}
\usepackage{algorithmicx}
\usepackage{amsthm}
\usepackage{amsmath}
\usepackage[hypcap=false]{caption}
\usepackage{tikz}
\usepackage{multicol}
\usepackage{hyperref}
\usepackage[framemethod=tikz]{mdframed}
\usepackage{marvosym}
\usepackage{amssymb}
\usepackage{stmaryrd}
\usepackage{wrapfig}

\usetikzlibrary{arrows,automata,positioning}
\tikzset{inner sep=0pt,minimum size=1pt}

\DeclareMathAlphabet{\mathbbm}{U}{bbm}{m}{n}
\newcommand{\IR}{\mathbbm{R}}     
\newcommand{\now}{\mathit{now}}
\renewcommand{\epsilon}{\varepsilon}

\newtheorem{theorem}{Theorem}
\newtheorem{lemma}{Lemma}
\theoremstyle{definition}
\newtheorem{definition}{Definition}
\surroundwithmdframed[outerlinewidth=0.4pt,
  innerlinewidth=0.4pt,
  middlelinewidth=0pt,
  middlelinecolor=white]{definition}

\newenvironment{smallenum}{
\begin{itemize}
    \setlength{\topsep}{-3pt} 
    \setlength{\partopsep}{-3pt}
  \setlength{\itemsep}{0pt}
  \setlength{\parskip}{-1pt}
  \setlength{\parsep}{-6pt}
}{\end{itemize}}

\begin{document}

\maketitle

    \begin{abstract}
        \noindent
        In this paper, we consider the problem of cross-chain payment whereby customers of different
        escrows---implemented by a bank or a blockchain smart contract---successfully transfer
        digital assets without trusting each other. Prior to this work, cross-chain payment problems
        did not require this success or any form of progress.  We introduce a new specification
        formalism called \emph{Asynchronous Networks of Timed Automata (ANTA)} to formalise such protocols.
        We present the first cross-chain payment protocol that ensures termination in a
        bounded amount of time and works correctly in the presence of clock skew.
        We then demonstrate that it is impossible to solve this problem without assuming
        synchrony, in the sense that each message is guaranteed to arrive within a known amount of time.
        We also offer a protocol that solves an eventually terminating variant of this cross-chain
        payment problem without synchrony, and even in the presence of Byzantine failures.
        \end{abstract}

        \section{Introduction}
        
        With the advent of various payment protocols comes the problem of interoperability between them. 
        A simple way for users of different protocols to interact is to do a \emph{cross-chain payment} whereby  
        intermediaries can help customer Alice transfer digital assets to Bob even
        though Alice and Bob own accounts in different banks or blockchains.  
        
        A payment between two customers of the same bank is simple.
        Alice just informs the bank that she wants to transfer a certain amount from her
        account to the account of the receiving party Bob; and then the bank carries out this request.
	Alice and Bob do not need to trust each other but need to trust the bank to not withdraw the
        money from Alice's account and never deposit it on Bob's.
        As Bob trusts the bank, he can issue a signed certificate assuring Alice that if the bank
        says that he has been paid, then Alice is of the hook, and any further disputes about
        possible non-payment to Bob are between Bob and the bank. To prevent litigation against her,
        Alice simply needs this statement from Bob as well as a statement that Bob has been paid. As
        Alice does not trust Bob this second statement must come from the bank.
        
        To generalise this protocol to payments between customers of different banks, it helps if the two
        banks have ways to transfer assets to each other, and moreover trust each other. 
        A sound protocol becomes: 
        \begin{smallenum}
        \item[(i)] Bob provides Alice with a signed statement that all he requires for her to have satisfied
          her payment obligation, is a statement from his own bank saying that he has been paid.
        \item[(ii)] Alice's bank promises Alice that if she transfers money to
          Bob, she will get a statement from Bob's bank that the transfer has been carried out.
        \item[1.] Alice orders her bank to initiate the transfer to Bob.
        \item[2.] Alice's bank withdraws the money from her account, and sends it to Bob's bank.
        \item[3a.] Bob's bank places the money in Bob's account
        \item[3b.] and notifies Alice's bank of this.
        \item[4.] Alice's bank forwards to Alice the statement by Bob's bank saying that Bob has been paid.
        \end{smallenum}
        This is roughly how payments between customers of different banks happen in the world of banking.
        Step (ii) is part of general banking agreements, not specific to Alice and Bob.
        Step (i) is typically left implicit in negotiations between Alice and Bob.
        All that Alice needs is the combination of steps (i), (ii) and 4 above.
        Once step (i) and (ii) have been made, she confidently takes step 1, knowing that this will be followed by Step 4.
        Alice's bank is willing to take step (ii) because it trusts Bob's bank, in the sense that step 2
        \emph{will} be followed by Step 3b.
        In fact, when the two banks trust each other, and have ways to transfer assets to each other,
        they can abstractly be seen as one bigger bank, and the problem becomes similar to the problem of payments 
        between customers of the same bank.\vspace{3.5ex}

  \noindent\begin{minipage}{\linewidth}
    \begin{center}
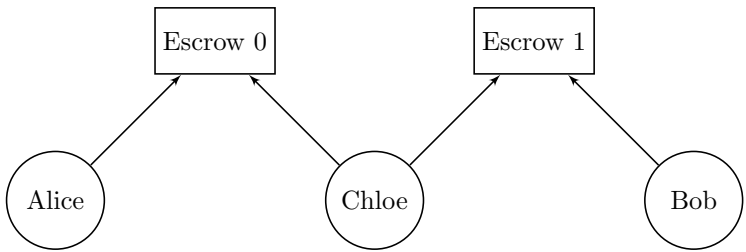

    \begin{tikzpicture}[xscale=0.45,yscale=0.75,>=latex', auto, semithick, node distance=3cm, inner sep=3pt ]
      
      \node[state, circle, minimum size=1.3cm]   (A)                        {Alice};
      \node[state, rectangle]           (B) [above right of = A]        {Escrow 0};
      \node[state, circle, minimum size=1.3cm]           (C) [below right of = B]        {Chloe};
      \node[state, rectangle]           (D) [above right of = C]        {Escrow 1};
      
      \node[state, circle, minimum size=1.3cm]           (E) [below right of = D]        {Bob};
    
      \path (A) edge [->]               node        {} (B)
            (B) edge  [<-]              node        {} (C)
            (C) edge   [->]             node        {} (D)
            (D) edge    [<-]            node        {} (E)
            ;
     
    \end{tikzpicture}
    \vspace{-1ex}
    \end{center}
    \captionof{figure}{Trust relations}
    \label{trust}
    \end{minipage}\vspace{3.5ex}
        
        The problem becomes more interesting when the banks cannot transfer assets to each other and the only trust is the one of
        customers to their own bank. Typical solutions consist of considering banks as \emph{escrows} and having intermediaries, 
        like Chloe, that play the role of connectors between these escrows.
        Figure~\ref{trust} depicts the relations of trust between three customers and two escrows and 
        where the flow of money is from left to right.
        Thomas and Schwartz~\cite{TS15} proposed two interledger protocols: (i)~the universal protocol requires synchrony~\cite{DLS88}
        in that every message between participants is received within a known upper bound and the clock skew
        between participants is bounded;  (ii)~the atomic protocol
        coordinates transfers using an ad-hoc group of notaries, selected by the participants.
        Herlihy, Liskov and Shrira~\cite{Her19} represent a cross-chain payment as a deal matrix $M$
        where $M_{i,j}$ characterises a transfer of some asset from participant $i$ to participant $j$. They offer a
        timelock\footnote{\url{https://en.bitcoin.it/wiki/Hashed_Timelock_Contracts}.} commit protocol that 
        requires synchrony, and a certified blockchain commit protocol that requires partial synchrony and a certified blockchain.
        The correctness proofs of these algorithms consist of observing reasonable properties that sum up to a correctness argument. 
        But, as far as we know, there is no formal treatment of a cross-chain payment protocol in the presence of (multiplicative) clock skew.
        
        In this paper, we introduce a new specification formalism for cross-chain payment protocols, called 
        \emph{Asynchronous Networks of Timed Automata (ANTA)}. ANTA simplify the representation of cross-chain payment
        to a customer automaton and an escrow automaton that describe states from which outgoing transitions are immediately 
        enabled and states from which outgoing transitions are conditional upon some predicates. These automata allow us to reason
        formally about the liveness and safety of cross-chain payment protocols. ANTA differ
        from Alur and Dill's \emph{timed automata}~\cite{AD94}, their 
        networks \cite{BLLPY96} and I/O automata~\cite{Lyn03} in subtle ways, tuned to the problem at hand.
        We illustrate ANTA by specifying the interledger universal protocol and proving that it solves the time-bounded variant of the 
        cross-chain payment problem.
        Moreover, we fine-tune the protocol to work correctly even in the presence of clock skew.
        
        We also show that there exist no algorithm that can solve the time-bounded cross-chain payment problem without synchrony
        even if all participants either behave  correctly or simply crash.  This impossibility result relies on classic indistinguishability 
        arguments  from the distributed computing literature and highlights an interesting relation
        between the the cross-chain payment problem and the well-known transaction commit problem~\cite{GL06,Had90,Gue95}.
        Inspired by this earlier work on the transaction commit, we define a weaker variant of our problem called the eventual 
        cross-chain payment problem that relaxes the liveness guarantees to be solvable with partial synchrony. This new problem
        differs from the transaction commit problem and its variants like the non-blocking weak atomic commit 
        problem~\cite{Gue95} by tolerating 
        Byzantine failures. It is also different from the interledger problem solved in~\cite{TS15},
        and from the problem solved by the certified blockchain commit protocol of \cite{Her19} in a partially synchronous setting,
        by requiring some liveness.
        In particular, a protocol where all participants always abort is not permitted by our problem specification.
        We propose an algorithm that solves the eventual cross-chain payment problem 
        only assuming partial synchrony, and in the presence of Byzantine failures, using the ANTA formalism.

        Interestingly, the classical notion of \emph{atomicity}, meaning that the entire transaction goes through, or
        is rolled back completely, is not appropriate for this kind of protocols. In the words of \cite{Her19},
        ``This notion of atomicity cannot be guaranteed when parties are potentially malicious:
        the best one can do is to ensure that honest parties cannot be cheated.''

        \section{Related work}
        
       \paragraph{The Interledger protocol.}
        In \cite{TS15} a protocol is presented for payments across payment systems.
        Here a \emph{payment system} can be thought of as an independent bank, where people can have
        accounts. The intended application is in digital payment systems, such as Bitcoin~\cite{Nak08}.
        In \cite{TS15} the payment systems, or rather the functionality they offer, are loosely
        referred to as \emph{ledgers}, and payments between customers of different ledgers as \emph{interledger payments}.
        Following popular terminology, we here speak of \emph{escrows} and \emph{cross-chain payments}.
        
        The protocol from \cite{TS15} generalises to the situation of a longer zigzag than indicated in
        Figure~\ref{trust}, involving $n$ escrows and $n{-}1$ intermediaries Chloe$_i$.
        It comes in two variants: an \emph{atomic mode}, in which transfers are coordinated using an ad-hoc
        group of notaries, selected by the participants, and a \emph{universal mode}, not requiring external coordination.
        The atomic mode relies on more than two-third of the notaries to be reliable, although it may be
        unknown which ones.
        
        The correctness proof by Thomas and Schwartz~\cite{TS15} of the universal mode of the protocol requires the assumption
        of \emph{bounded synchrony} on the communication between the parties in the protocol, simply called
        \emph{synchrony} by Dwork et al.~\cite{DLS88}. This assumption says that any message sent by one of the parties,
        is guaranteed to be received by the intended recipient(s). Moreover, there is a known upper bound on
        the time a message is underway. Additionally, there needs to be a known upper bound on the clock skew
        between participants. Here we point out that this assumption is necessary, in the sense that without
        this assumption, and without the help of mostly reliable notaries, no correct cross-chain payment protocol exists.
        
        To make such a statement, we need to define when we consider a cross-chain payment protocol correct.
        The correctness proof by Thomas and Schwartz~\cite{TS15} consists of reasonable properties that are shown to hold for the chosen protocol.
        These properties are stated in terms of that specific protocol, and
        the reader can infer that they sum up to a correctness argument.
        But no formal statement occurs of what it means for a general cross-chain payment protocol to be correct,
        and this is what we need to make any negative statement about it.

        \paragraph{Other cross-chain technologies.}
        
        In the lightning network, Poon and Dryja~\cite{PT16} can relay payment outside the
        blockchain or ``offchain'' through connected intermediaries, but they require
        synchrony and do not consider clock skew.
                Gazi, Kiayias and Zindros \cite{Gazi18} propose a rigorous formalisation of ledger
        and cross-chain transfers, but focus on proof-of-stake (in particular  Cardano
        Ouroboros). The assumption of ``semi-synchronous communications", which in practice
        corresponds to synchrony as defined by Dwork, Lynch and Stockmeyer \cite{DLS88},
        make their results not extensible to partial synchrony. They do not consider clock skew.
                Herlihy proposes atomic cross-chain swaps~\cite{Her18} to exchange assets between distinct
        blockchains, but only in a synchronous environment without clock skew.
	Ron van der Meyden~\cite{Mey19} verifies a cross-chain swap protocol by modelling a
        timelock predicate as a Boolean variable indicating whether the asset is transferred. This
        approach also requires synchrony, and does not consider clock skew. In \textsc{XClaim}~\cite{ZHLPGK19}
        Zamyatin et al.\ propose a solution to swap blockchain-backed assets. Their protocol assumes
        that adversaries are behaving rationally, and requires synchrony.
         
        Lind et al.~\cite{LNEKPS18} relax the synchrony assumption but require a trusted execution
        environment (TEE). Such a solution cannot be used to solve our problem as it would require
        trusting a third-party, often represented as the manufacturer of this TEE.
        
        Other approaches \cite{Woo16,RG19,Her19} rely on a separate blockchain that plays the same role as our transaction manager
        (cf.\ Section~\ref{sec:transaction manager}).
        However, \cite{Woo16} and \cite{RG19} do not aim at ensuring liveness,
        and \cite{Her19} aims at ensuring liveness only in periods where communication proceeds synchronously.
        Wood~\cite{Woo16} proposes a multi-chain solution that aims at combining heterogeneous 
        blockchains together without trust. As far as we know, it has not been proved that the protocol terminates.
        Ranchal-Pedrosa and Gramoli~\cite{RG19} relax the synchrony assumption using an alternative `child' blockchain
        to the so-called `parent' blockchain in order to execute a series of transfers outside the parent blockchain.
        This protocol does not guarantee that the set of intermediary transfers on the child blockchain eventually take effect.
        Herlihy, Liskov and Shrira~\cite{Her19} model cross-chain deals as a matrix $M$ where 
        $M_{i,j}$ characterises a transfer of some asset from participant $i$ to participant $j$,  and
        offer a  timelock-based solution under the synchrony assumption, without clock skew,
        and a certified blockchain commit protocol that requires partial synchrony and a certified blockchain. 
        As remarked in \cite{Her19}, a strong liveness guarantee is not feasible when merely
        assuming partial synchrony. In this context the targeted cross-chain deal problem admits solutions where all correct
        processes simply abort.  Our corresponding problem differs by requiring a weaker liveness guarantee in that
        it formulates conditions under which a successful transfer is ensured.
        We present a more detailed comparison between our work and that of \cite{Her19} in Section~\ref{deals}.
        
        \paragraph{Crash fault tolerant solutions.}

        The transaction commit problem is a classical problem from the database literature, tackled
        for instance by Gray \& Lamport \cite{GL06}, Guerraoui \cite{Gue95} and
        Hadzilacos \cite{Had90}. It consists of ensuring that either all the processes commit a
        given transaction or all the processes abort this transaction.

        The Non-Blocking Atomic Commitment problem has been formally defined by Guerraoui, for
        asynchronous systems with unreliable failure detectors (encapsulating partial synchrony) in
        a crash (fail-stop) model \cite{Gue95} but not in a Byzantine model.
        Guerraoui proved an impossibility result when the problem requires that
        ``If all participants vote \emph{yes}, and there is no failure, then every correct participant eventually decides
        \emph{commit}''. He then proposes a weaker variant of the problem where the above
        requirement is replaced by ``If all participants vote \emph{yes}, and no participant is ever
        suspected, then every correct participant eventually decides \emph{commit}''. The two problems we
        consider in this paper present a similar distinction; however, our problems target Byzantine
        fault tolerance, not crash fault tolerance.

        Anta, Georgiou and Nicolaou \cite{Anta18} propose a general and rigorous definition of a
        ledger in which they consider the atomic append problem in an asynchronous model. Similarly, their model considers exclusively crash failures.
       
        \section{Model and definitions}\label{correctness}

        \subsection{Participants and money}
        
        We assume $n$ banks or \textit{escrows} $e_0,\dots,e_{n-1}$ and $n{+}1$ \textit{customers} $c_0,\dots,c_{n}$.
        These $2n{+}1$ processes are called \textit{participants}.

        \paragraph{Escrows.} An escrow is a specific type of process that can handle values for other parties in a predefined manner. 
        
        \paragraph{Customers.} Customer $c_0$ is Alice
        and $c_n$ is Bob. The customers $c_1,\dots,c_{n-1}$ are intermediaries in the interaction between
        Alice and Bob; we call them \textit{connectors}, named\ Chloe$_i$.

        \paragraph{Trust.}
        Customers $c_{i-1}$ and $c_{i}$ have accounts at escrow $e_{i-1}$, and trust this escrow ($i=1,\dots,n$).
        We do not assume any other relations of trust.
        \vspace{4ex}

        \noindent\begin{minipage}{\linewidth}
          \begin{center}
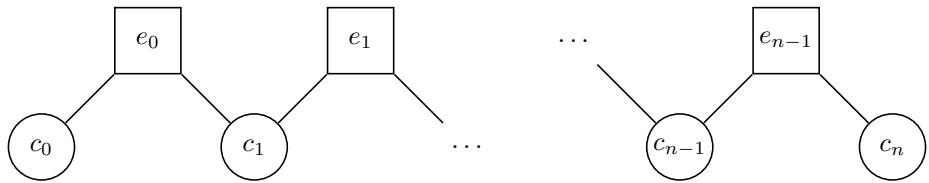

          \begin{tikzpicture}[xscale=0.45,yscale=0.75,>=latex', auto, semithick, node distance=2cm]
            
            \node[state, circle]   (A)                        {$c_0$};
            \node[state, rectangle]           (B) [above right of = A]        {$e_0$};
            \node[state, circle]           (C) [below right of = B]        {$c_1$};
            \node[state, rectangle]           (D) [above right of = C]        {$e_1$};
            
            \node[state, draw=white]           (E) [below right of = D]        {$\ldots$};
            \node[state, draw=white]           (K) [above right of = E]        {$\ldots$};
            
              \node[state, circle]   (F)          [below right of = K]              {$c_{n-1}$};
            \node[state, rectangle]           (G) [above right of = F]        {$e_{n-1}$};
            \node[state, circle]           (H) [below right of = G]        {$c_n$};
          
            \path (A) edge                node        {} (B)
                  (B) edge                node        {} (C)
                  (C) edge                node        {} (D)
                  (D) edge                node        {} (E)
                  (K) edge                node        {} (F)
                  (F) edge                node        {} (G)
                  (G) edge                node        {} (H)   
                  ;
           
          \end{tikzpicture}
          \end{center}
          \captionof{figure}{Customers and escrows.}
          \label{fig:topology}
          \end{minipage}

        \paragraph{Topology.} Not every customer can send value to any other.
        Here we assume that value can be transferred directly only between customers of the same escrow.
        Moreover, any transfer between two customers of an escrow can be modelled as two transfers:
        one from the originating customer to the escrow, and one from the escrow to the receiving customer.
        Thus, the connections from Figure~\ref{fig:topology} describe both the relations of trust and the
        possible transfers of value. The case $n=2$ was depicted in Figure~\ref{trust}.

        We can talk of the \textit{customers of an escrow} to designate the customers that
        are connected to this escrow. We also talk of the \textit{escrows of a customer} to
        designate the escrows that are connected to this customer---only one escrow for
        Alice and Bob, and two escrows for a connector.
        
        \paragraph{Placing value in escrow.}
        Customers can send a specific type of message to ask their escrow to put money aside for them.
        In particular, two customers may make a deal with an escrow to place value originating from
        the first customer ``in escrow'', and, after a predefined period, depending on which
        conditions are met, either complete the transfer to the second customer, or return the value
        to the first one.
        
        \paragraph{Abstracting the transfer of value.} There are many ways of transferring value from one party to another: one could give a physical object, such as cash or gold, to someone else. One could also send a transaction to transfer cryptocurrency or tokens on a distributed ledger, or send a specific message on some banking application. We do not care of how this process is implemented, and we suppose that the participants have already agreed upon the value they expect to be transferred. 

        We use therefore a unified notation: $s(p, \$)$ to say ``send a message to trigger the
        transfer of some previously agreed-upon value to participant $p$''.
        
        \paragraph{Chloe's fee.}
        As Chloe helps out transferring value from Alice to Bob, it is only reasonable that she is
        paid a small commission. Hence the value transferred from Alice to Chloe might be larger than the
        value transferred from Chloe to Bob. Additionally, these values may be expressed in
        different currencies, with possibly fluctuating exchange rates, or they may be objects such
        as bags of flour that have a quality-dependent value. Deciding which values to transfer may
        thus be an interesting problem. However, it is entirely orthogonal to the matter discussed
        in this paper, and hence we shall not consider it any further.

        \subsection{Communication and computation model}

        The following model holds for the rest of the paper. When necessary, we will mention
        explicitly if we have to add some assumptions, such as synchrony of communication.

        \paragraph{Communication.} We assume that the network does not lose, duplicate, modify or
        create messages. However messages can be delayed for arbitrarily long: by default we assume
        asynchronous communication.
        
        \paragraph{Authentication.}

        We assume that each customer can sign a message with his unique
        identifier, thanks to an idealised public key infrastructure. No other process can
        forge its signature, and any process (including escrows) can verify it. An escrow
        is not necessarily able to sign a message.

        \paragraph{Certificates.} As a consequence, any customer can issue a certificate by signing
        an appropriate message. For example, Bob can issue a
        receipt certificate to Alice by signing a \textsc{received} message. By combining several
        signatures, one can define threshold certificates, for instance requiring
        the signature of a \textsc{commit} message by strictly more than $n/3$ customers
        or notaries. A correct implementation of certificates should take care of preventing replay attacks.

        \paragraph{Faults.} We do not make any assumption on the behaviour of the processes
        \textit{a priori}. Later we will define a protocol, and processes will either follow the
        protocol or deviate from the protocol.

        \subsection{Synchronous versus asynchronous communication}

        In the literature five levels of synchrony in communication can be distinguished. As indicated in the table
        below, terminology is not uniform between the concurrency and the distributed systems communities.
        \begin{center}
        \begin{tabular}{|c|c|}
          \hline
          \emph{concurrency} & \emph{distributed systems}\rule{0pt}{10pt}\\
          \hline
          \hline
          synchronous communication & rendezvous\rule{0pt}{10pt} \\
          \hline
          \multicolumn{2}{|c|}{I/O automata\rule{0pt}{10pt}} \\
          \hline
          asynchronous communication & synchronous communication\rule{0pt}{10pt} \\
          \cline{2-2} 
          & partially synchronous communication\rule{0pt}{10pt} \\
          \cline{2-2} 
          & asynchronous communication\rule{0pt}{10pt} \\
          \hline
        \end{tabular}
        \end{center}
        In the concurrency community, communication is called \emph{synchronous} if sending\pagebreak
        and receiving occur simultaneously, and the sender cannot proceed before receipt of the
        message is complete. This is the typical paradigm in process algebras such as CCS \cite{Mi89}.
        Communication is called \emph{asynchronous} if sending occurs strictly before receipt, and
        the sender can proceed after sending regardless of the state of the recipient(s).
        An intermediate form is modelled by I/O automata~\cite{Lyn03}; here sending and receipt is
        assumed to occur simultaneously, yet the sender proceeds after sending regardless of the
        state of possible receivers.

        In the distributed systems community all communication is by default assumed to be
        asynchronous in the sense above; synchronous communication as defined above is
        sometimes called a \emph{rendezvous}. Following \cite{DLS88}, communication is called
        \emph{synchronous} when there is a known upperbound on the time messages can be in transit,
        and moreover there is a known upperbound on the relative clock skew between parallel processes.
        It is \emph{partially synchronous} when these upperbounds exist but are not known, or when
        it is known that after a finite but unknown amount of time these upperbounds will come into effect.
        If these conditions are not met, communication is deemed \emph{asynchronous}.

        The present paper follows the latter terminology; we speak of \emph{fully synchronised}
        communication when referring to synchronous communication from concurrency theory.
        
        \subsection{Informal description of the syntax and semantics of ANTA}

        Although strongly inspired by the timed automata of
        Alur and Dill \cite{AD94} and their networks \cite{BLLPY96}, our ANTA differs from those models in
        subtle ways, tuned to the problem at hand. The first A refers to asynchronous communication
        in the concurrency-theoretic sense, and contrasts with the fully synchronised communication
        that is assumed in NTA \cite{BLLPY96}.
        
        In Figure~\ref{automata} our time-bounded cross-chain protocol---essentially the
        universal mode of the interledger protocol from \cite{TS15}---is depicted as an ANTA\@.
        There is one automaton for each participant in the protocol, that is, for each escrow $e_i$
        ($i=1,\dots,n{-}1$) and each customer $c_i$ ($i=1,\dots,n$). Each automaton is equipped with a
        unique identifier, in this case $e_i$ and $c_i$.
        Each automaton has a finite number of \emph{states}, depicted as circles, one of which is marked as
        the \emph{initial state}, indicated by a short incoming arrow. The states are partitioned into
        \emph{termination states}, indicated by a double circle, \emph{input states}, coloured white, and
        \emph{output states}, coloured grey or black. Furthermore there are finitely many \emph{transitions},
        indicated as arrows between states. The transitions are partitioned into \emph{input transitions},
        labelled $r(id,m)$, \emph{output transitions}, labelled $s(id,m)$, and \emph{time-out transitions},
        labelled by arithmetical formulas $\psi$ featuring the variable $\now$. Here $id$ must be the
        identifier of another automaton in the network, and $m$ a message, taken from a set {\tt MSG} of
        allowed messages.
        Whereas each input and time-out transition has a unique label $r(id,m)$ and $\psi$, respectively,
        an output transition may have multiple labels $s(id,m)$.
        All transitions may have additional labels of the form $x:=\now$ for some variable $x$.
        A termination state has no outgoing transitions.
        An output state has exactly one outgoing transition, which much be an output transition.
        An input state may have any number of outgoing input and time-out transitions, and no outgoing output transitions.
        
        Each automaton keeps an internal clock, whose value, a real number, is stored in the variable $\now$.
        The value of $\now$ increases monotonically as time goes on. All variables maintained by an automaton
        are local to that automata, and not accessible by other automata in the network.
        Each transition is assumed to occur instantaneously, at a particular point in time.
        In case a transition occurs that is labelled by an assignment $x:=\now$,
        the variable $x$ will remember the point in time when the transition took place.
        Such a variable may be used later in time-out formulas.
        When (or shortly before, see below) an output transition with label $s(id,m)$ occurs, the automaton
        sends the message $m$ to the automaton with identifier $id$.
        A time-out transition labelled $\psi$ is \emph{enabled} at a time $\now$ when the formula $\psi$
        evaluates to {\tt true}. An input transition labelled $r(id,m)$ is \emph{enabled} only at a time when the
        automaton receives the message $m$ from the automaton $id$ in the network.
        Whereas an output transition may be scheduled to occur by the automata at any time, an input or
        time-out transition can occur only when enabled.
        
        When an automaton is not performing a transition, it must be in exactly one of its states.
        It starts at the initial state, where its clock is initialised with an arbitrary value.
        When the automaton is in an input state, it stays there (possibly forever) until one of its outgoing
        transitions becomes enabled; in that case that transition will be taken immediately.
        In case multiple transitions become enabled simultaneously, the choice is non-deterministic.
        When the automaton reaches a termination state, it halts.
        
        In general, an output state is labelled with a positive \emph{time-out} value $\mathit{to}\in\IR\cup\{\infty\}$.
        It constitutes a strict upperbound on the time the automaton will stay in that state.
        In case the automaton enters an output state at time $\now$, it will take
        its outgoing transition between times $\now$ and $\now+\mathit{to}$.
        In case its output transition has multiple labels $s(id,m)$,
        the corresponding transmissions need not occur simultaneously; they can occur in any order between
        $\now$ and $\now+\mathit{to}$. The output transition is considered to
        be taken when the last of these actions occurs.
        In this paper time-out values are indicated by colouring: for the grey states it is the constant
        $\epsilon$ from Section~\ref{reaction speed}, while for the black state it is $\infty$.

        \subsection{Cross-chain payment protocol}

        A cross-chain payment protocol prescribes a behaviour for each of the participants in the protocol, the
        escrows and the customers. 
        Let $\chi$ be a certificate signed by Bob saying that Alice's obligation to pay him has been met.

        \begin{definition}[Time-bounded cross-chain payment protocol]\label{time-bounded correct}\vspace{-3pt}
          A cross-chain payment protocol is a \textit{time-bounded cross-chain payment protocol} if it satisfies the following properties:
          \begin{itemize}
            \item[C] \textbf{Consistency.} For each participant in the protocol it is possible to abide by the protocol.
            \item[T] \textbf{Time-bounded termination.} Each customer that abides by the
              protocol, and either makes a payment or issues a certificate, terminates within an
              a priori known period, provided her escrows abide by the protocol.
            \item[ES] \textbf{Escrow security.} Each escrow that abides by the protocol does not loose money.
            \item[CS] \textbf{Customer security.}
              \begin{itemize}
                \item[CS1] Upon termination, if Alice and her escrow abide by the protocol, Alice  has either got her money back or received the statement $\chi$.
                \item[CS2] Upon termination, if Bob and his escrow abide by the protocol, Bob has either received the money or not issued certificate $\chi$.
                \item[CS3] Upon termination, each connector that abides by the protocol has got her money back,
                provided her escrows abide by the protocol.
              \end{itemize}
            \item[L] \textbf{Strong liveness.} If all parties abide by the protocol, Bob is paid eventually.
          \end{itemize}
        \end{definition}

        Requirement C (\emph{consistency} of the protocol) is essential.
        In the absence of this requirement, any protocol that prescribes an impossible task for each
        participant would be a correct cross-chain payment protocol (since it trivially meets
        T, ES, CS and L).
        
        Requirements ES and CS (the \emph{safety} properties) say that if a participant abides by the
        protocol, nothing really bad can happen to her. These requirements do not assume that any other
        participant abides by the protocol, and should hold no matter how malicious the other participants
        turn out to be. The only exception to that is that the safety properties for a customer (CS)
        are guaranteed only when the escrow(s) of this customer abide by the protocol.
        
        Property L, saying that the protocol serves its intended purpose, is the only one that is
        contingent on \emph{all} parties abiding by the protocol.

        \section{A time-bounded cross-chain payment protocol}

        \subsection{Assumptions}\label{bounded synchrony}

        \paragraph{Synchrony.}

        The assumption of \emph{synchrony} considered e.g.\ by \cite{DLS88}, and called \emph{bounded synchrony} by \cite{TS15},
        says `that there is a fixed upper bound $\Delta$ on the time for messages to be delivered
        (\emph{communication is synchronous}) and a fixed upper bound $\Phi$ on the rate at which one
        processor's clock can run faster then another's (\emph{processors are synchronous}), and that these
        bounds are known a priori and can be ``built into'' the protocol.' \cite{DLS88}\linebreak[3]
        A consequence of this assumption is that if participant $p_1$ sends at its local time $t_0$
        a message to participant $p_2$, and participant $p_2$ takes $t$ units of its local time to send an
        answer back to $p_2$, then $p_1$ can count on arrival of that reply no later than time
        $t_0+ \Phi\cdot t + 2\cdot\Delta$.

        \paragraph{Bounded reaction speed.}\label{reaction speed}
        
        When a participant in the protocol receives a message, it will take some time to calculate the right
        response and then to transmit that response. Here it will be essential that that amount of
        time is bounded. So we assume a  reaction time $\epsilon > 0$ such that any received message
        can be answered within time $\epsilon$.

        \subsection{A cross-chain payment protocol formalised as an ANTA}

        In this section we formally model the \textit{universal mode} of the interledger
        protocol from~\cite{TS15} as an ANTA\@.
        Moreover, we replace the timing constants employed in \cite{TS15} by parameters, and
        then calculate the optimal value of these parameters to unsure correctness of the protocol
        in the presence of clock skew.

        \begin{figure}
          \footnotesize
          \input{automata}
          \centerline{\box\graph}
          \caption{Automata representing escrows and customers}
          \label{automata}
          \end{figure}
        
        To interpret Figure~\ref{automata}, all that is left to do is specify the messages that are
        exchanged between escrows and their customers. We consider 4 kinds of messages.
        One is the certificate $\chi$, signed by Bob, saying that Alice's obligation to pay him has been met.
        Another is the value $\$$ that is transmitted from one participant to another.
        The remaining messages are promises made by escrow $e_i$ to its customers $c_i$ and $c_{i+1}$, respectively:%
        \vspace{1ex}

        $G(d) :=$ ``\parbox[t]{5in}{I guarantee that if I receive $\$$ from you at my local time $w$,\\
                      then I will send you either $\$$ or $\chi$ by my local time $w+d$.''}\vspace{1ex}
        
        $P(a) :=$ ``\parbox[t]{5in}{I promise that if I receive $\chi$ from you at my local time $v$, with $v<\now+a$,\\
                      then I will send you $\$$ by my local time $v+\epsilon$.''}\vspace{1ex}
        
        The automata of Figure~\ref{automata} can be informally described as follows:
        A escrow $e_i$ first send promise $G(d_i)$ to its (upstream) customer $c_i$.
        Here ``upstream'' refers to the flow of money.
        The precise values of $d_i$ will be determined later; here they are simply parameters in
        the design of the protocol. Then it awaits receipt of the money/value from customer $c_i$.
        If the money does arrive, the escrow issues promise $P(a_i)$ to its downstream customer
        $c_{i+1}$ as soon as it can.
        It remembers the time this promise was issued as $u$.
        Then it awaits receipt of the certificate $\chi$ from customer $c_{i+1}$.
        If the certificate does not arrive by time $u+a_i$, a time-out occurs, and the escrow refunds the money to
        customer $c_i$. If the money does arrive in time, the escrow reacts by forwarding the certificate to
        customer $c_i$, and forwarding the money to customer $c_{i+1}$.
        
        A connector Chloe$_i$ starts by awaiting promises $G(d_i)$ from her downstream escrow $e_i$, and
        $P(a_{i-1})$ from her upstream escrow $e_{i-1}$. Then she proceeds by sending the money to
        escrow $e_i$. After sending the money, Chloe$_i$ waits for escrow $e_i$ to send her either the
        certificate $\chi$ or the money back. In the latter case, her work is done; in the former, she
        forwards the certificate to escrow $e_{i-1}$ and awaits for the money to be send by escrow $e_{i-1}$.
        
        The automata for Alice and Bob are both simplifications of the one for Chloe$_i$.
        Alice awaits promise $G(d_0)$ from her escrow, and then sends the
        escrow the money. The protocol allows her to wait arbitrary long before taking that step.
        Subsequently, she patiently await for either the return of her money, or certificate $\chi$.
        Bob awaits promise $P(a_{n-1})$ from his escrow, and then issues certificate $\chi$ and sends it to
        his escrow. He then awaits the money.
        
        \subsection{Running the protocol}\label{run}
        
        The protocol consists of two parts. The \emph{set-up} involves the sending and receiving of the
        promises $G(d_i)$. As these promises are not time-sensitive, they can be exchanged months before the
        \emph{active part} of the protocol is ran, consisting of all other actions.
        Here an \emph{action} is an entity $\textit{act}\color{blue} \MVAt p$, with $\textit{act}$ a transition
        label, and $p$ the identifier of the participant taking that transition.
        The active part has essentially only one successful run, i.e., when never taking a time-out transition,
        consisting of the following actions, executed in the following order. Actions separated by
        commas can be executed in either order.
        \begin{center}
        \begin{tabular}{@{}ll@{}c@{\hspace{-6.25pt}}rr@{}}
        $s(e_0,\$) \color{blue} \MVAt c_0$ & $r(c_0,\$)\color{blue}\MVAt  e_0$ &&
        $s(c_1,P(a_0)) \color{blue}\MVAt e_0$ & $r(e_0,P(a_0)) \color{blue} \MVAt c_1$ \\
        $s(e_1,\$) \color{blue} \MVAt c_1$ & $r(c_1,\$)\color{blue}\MVAt  e_1$ &&
        $s(c_2,P(a_1)) \color{blue}\MVAt e_1$ & $r(e_1,P(a_1)) \color{blue} \MVAt c_2$ \\
        \dots \\
        $s(e_{n-1},\$) \color{blue} \MVAt c_{n-1}$ & $r(c_{n-1},\$)\color{blue}\MVAt  e_{n-1}$ &&
        $s(c_n,P(a_{n-1})) \color{blue}\MVAt e_{n-1}$ & $r(e_{n-1},P(a_{n-1})) \color{blue} \MVAt c_n$ \\
        $s(e_{n-1},\chi) \color{blue} \MVAt c_n$ & $r(c_n,\chi)\color{blue}\MVAt  e_{n-1}$ &
        $s(c_{n},\$) ,~s(c_{n-1},\chi) \color{blue}\MVAt e_{n-1}$ &
        $r(e_{n-1},\$) \color{blue} \MVAt c_{n}$ & ,\hfill$r(e_{n-1},\chi) \color{blue} \MVAt c_{n-1}$ \\
        \dots \\
        $s(e_{1},\chi) \color{blue} \MVAt c_2$ & $r(c_2,\chi)\color{blue}\MVAt  e_{1}$ &
        $s(c_{2},\$), s(c_{1},\chi) \color{blue}\MVAt e_{1}$ &
        $r(e_{1},\$) \color{blue} \MVAt c_{2}$ & ,\hfill$r(e_{1},\chi) \color{blue} \MVAt c_{1}$ \\
        $s(e_{0},\chi) \color{blue} \MVAt c_1$ & $r(c_1,\chi)\color{blue}\MVAt  e_{0}$ &
        $s(c_{1},\$), s(c_{0},\chi) \color{blue}\MVAt e_{0}$ &
        $r(e_{0},\$) \color{blue} \MVAt c_{1}$ & ,\hfill$r(e_{0},\chi) \color{blue} \MVAt c_{0}$ \\
        \end{tabular}
        \end{center}

        \subsection{Initialisation}\label{initialisation}
        
        There is a scenario where Chloe$_i$ will never send money to escrow $i$, namely when she receives
        promise $P(a_{i-1})$ from escrow $e_{i-1}$ before she receives promise $G(d_i)$ from escrow $e_i$.
        In that case the receipt of $P(a_{i-1})$ does not trigger a transition, and Chloe$_i$ will remain
        stuck in her second state. We now modify the protocol in such a way that that scenario can not
        occur. This can be done in three ways; it does not matter which of the three modifications we take.
        \begin{enumerate}
        \item One solution is to make Alice wait before starting the active part of the protocol (by leaving
          the black state) until a point in time when she is sure that all parties Chloe$_i$ already have
          received promise $G(d_i)$. If we assume that all parties start at the same time, using the
          reasoning of Section~\ref{bounded synchrony}, Alice has to wait
          at most $\Phi\cdot\epsilon + \Delta$ before this point has been reached.
          The only drawback of this solution is that it may be hard to realise that all parties start at the same time.
        \item Another approach is to assume that the set-up phase occurred long before Alice actually wants
          to send money to Bob. It may be part of a general banking agreement. Possibly each escrow always
          offers promises $G(d)$ for different values of $d$, and when sending money to escrow $i$, Chloe$_i$
          simply tags it as taking advantage of promise $G(d_i)$. In this approach, the protocol does not
          feature the transitions labelled $s(c_i,G(d_i))$ and $r(e_i,G(d_i))$, with the initial states
          shifted accordingly. Still, the promise $G(d_i)$ counts as having been made to customer $c_i$ by
          escrow $e_i$.
        \item A final solution is to introduce to the protocol a message ``We are ready'', sent by
          Chloe$_{n-1}$ to escrow $e_{n-2}$, and forwarded, via Chloe$_i$ and escrow $e_{i-1}$ all the way
          to Alice. Each Chloe$_i$ forwards the ``We are ready'' message only after receiving promise
          $G(d_i)$ from escrow $e_i$, so when Alice receives the ``We are ready'' message she can safely
          initiate the transfer. 
        \end{enumerate}

        \subsection{Correctness of the protocol}
        
        Now we show that the protocol from Figure~\ref{automata} is correct, in the sense that it satisfies the properties of
        Definition~\ref{time-bounded correct}, when making the assumptions of Section~\ref{bounded synchrony}.
        In doing so, we also calculate the values of the parameters $d_i$ and $a_i$.
        
        \paragraph{Consistency.}
        To check that the protocol is consistent, in the sense that each participant can abide by it,
        we first of all invoke the assumption of bounded reaction speed, described in Section~\ref{reaction speed},
        and use that the constant $\epsilon$  assumed to exist in Section~\ref{reaction speed} is in fact
        the time-out value associated to most output states. This ensures that it is always possible to send
        messages in a timely manner.
        
        The only remaining potential failure of consistency is when the protocol prescribes the transmission
        of a resource that it is not available. Assuming that the sending of promises and money is not an
        obstacle (Chloe has been selected, in part, for having this kind of money available), the only issue
        could be the sending of the certificate signed by Bob. For anyone but Bob this can only be done after
        receiving it first. However, a simple inspection of the automata of the escrows, Chloe$_i$ and
        Alice shows that any transition sending the certificate is preceded by a transition receiving it.
        This establishes requirement C\@.
        
        \paragraph{Escrow security.}
        
        That escrows cannot loose money (requirement ES) follows immediately from the observation that an
        escrow spends the money only after receiving it. This follows from the order of the
        transitions in the automaton for the escrows.
        
        \paragraph{Honesty.}

        Although not part of Definition~\ref{time-bounded correct}, we show that
        an escrow that issues a promise always keeps that promise, when abiding by the protocol.
        This property (H) will be be used below to establish CS.
        
        To show H, suppose the escrow $e_i$ issues promise $P(a_i)$, and subsequently
        receives the certificate $\chi$ from customer $c_{i+i}$ at a time $v<u+a_i$, where $u$ refers to the
        time promise $P(a_i,\epsilon)$ was issued. Then it is too soon for the time-out transition, so the transition
        labelled $r(c_{i+1},\chi)$ in the automaton of $e_i$ will be taken, at time $v$. The automaton shows
        that $s(c_{i+1},\chi)$ will occur by time $v+\epsilon$, thus fulfilling the promise.

        To show that an escrow that issues promise $G$ always keeps it, when abiding by the protocol,
        suppose the escrow $e_i$ receives the money at a time $w$. Then the transition
        labelled $r(c_i,\$)$ in the automaton of $e_i$ will be taken, at time $w$. The automaton shows that
        either $s(c_i,\$)$ will occur by time $w+\epsilon+a_i+\epsilon$, or 
        $s(c_i,\chi)$ will occur by time $w+\epsilon+a_i+\epsilon$.
        Thus, to guarantee that promise $G$ is met, we need to choose $d_i$ and $a_i$ in such a way that\vspace{-1.2ex}
        \begin{equation}\label{escrow}
        d_i \geq a_i+2\epsilon \vspace{.7ex}
        \end{equation}
        for $i=0,\dots n{-}1$.
        In fact, making the promise as strong as possible yields $d_i:=a_i+2\epsilon$.
        When this condition is met, we have established requirement H\@.

        \paragraph{Customer security and time-bounded termination.}
        We will check time-bounded termination (T) together with customer security (CS).
        To check requirement CS1, suppose that Alice will make the payment $s(e_0,\$)$, at time $t$.
        Then earlier she has received promise $G(d_0)$ from escrow $e_0$.
        This promise ensures Alice that escrow $e_0$
        will send her either $\$$ or $\chi$ by its local time $w+d_0$, where $w$ is the time Alice's payment
        is received. Consequently, by the reasoning of Section~\ref{bounded synchrony}, using the assumption of bounded synchrony,
        Alice will receive either certificate $\chi$ or her money back by time
        $t+\Phi\cdot d_0+2\cdot \Delta$.
        
        To check requirement CS2, suppose that Bob issues certificate $\chi$, at time $x$.
        Then earlier, at time $t$, he has received promise $P(a_{n-1})$ from escrow $e_{n-1}$.
        Moreover, $x < t+\epsilon$.
        For the promise to be meaningful, his certificate needs to arrive at
        $e_{n-1}$ before time $u+a_{n-1}$, where $u$ refers to the local time at $e_{n-1}$ when the promise was issued.
        By the reasoning of Section~\ref{bounded synchrony}, using the assumption of bounded synchrony,
        Bob's certificate will arrive at escrow $e_{n-1}$ before time $u +\Phi\cdot\epsilon+2\cdot\Delta$.
        Hence, we need to choose $a_{n-1}$ in such a way that
        \begin{equation}\label{Bob}
        a_{n-1}\geq\Phi\cdot\epsilon+2\cdot\Delta\;.
        \end{equation}
        When this requirement is met, the promise ensures Bob that escrow $e_{n-1}$ will send him the money
        by its local time $v+\epsilon$, where $v$ is the time Bob's certificate is received by $e_{n-1}$.
        Consequently, Bob will receive payment by time $x+\Phi\cdot \epsilon+2\cdot \Delta$.
        
        \begin{wrapfigure}[16]{r}{0.5\textwidth}
          \vspace{-4ex}
          \input{Chloe}
          \hfill\box\graph
        \end{wrapfigure}
        To check requirement CS3, suppose that Chloe$_i$ will make the payment $s(e_i,\$)$, at time $t_0$.
        Then earlier, she has received promise $G(d_i)$ from escrow $e_i$ and promise $P(a_{i-1})$
        from escrow $e_{i-1}$, the latter at time $t$. Moreover, $t_0< t+\epsilon$.
        Promise $G(d_i)$ ensures Chloe$_i$ that escrow $e_i$
        will send her either $\$$ or $\chi$ by its local time $w+d_i$, where $w$ is the time Chloe$_i$'s payment
        is received. Consequently, $c_i$ will receive either certificate $\chi$ or her money back by time
        $t_0+\Phi\cdot d_i+2\cdot \Delta$.
        
        Continuing with the case that she receives $\chi$ rather then her money back, she will forward
        $\chi$ to escrow $e_{i-1}$ by time $t_0+\Phi\cdot d_i+2\cdot \Delta + \epsilon$, which is before
        $t+\epsilon+\Phi\cdot d_i+2\cdot \Delta + \epsilon$.
        Hence it arrives at $e_{i-1}$ by its local time $u+2\cdot\Phi\cdot\epsilon+\Phi\cdot d_i+4\cdot \Delta$,
        where $u$ is the time promise $P(a_{i-1})$ was issued.
        Here we use that $\Delta$ is a valid upperbound on transition times by anyone's clock, and
        that there is no need to square $\Phi$ in $\Phi\cdot d_i$, as also the clock skew between escrows
        $e_i$ and $e_{i-1}$ is bounded by $\Phi$. This calculation is illustrated by the
          message sequence diagram above. Since $\chi$ needs to arrive at $e_{i-1}$ before time
        $u+a_{i-1}$ in order for promise $P(a_{i-1})$ to be meaningful, we need to pick
        \begin{equation}\label{alpha}
        a_{i-1}\geq 2\cdot\Phi\cdot\epsilon+\Phi\cdot d_i+4\cdot \Delta
        \end{equation}
        for $i=1,\dots,n{-}1$.
        When (\ref{alpha}) holds, promise $P(a_{i-1})$ ensures Chloe$_i$ that escrow $e_{i-1}$ will
        send her the money by its local time $v+\epsilon$, where $v$ is the time the
        certificate is received by $e_{i-1}$.  Consequently, $c_i$ will receive $\$$ by time
        $t_0+\Phi\cdot d_i+4\cdot \Delta + \epsilon + \Phi\cdot\epsilon$.
        Hence, assuming (\ref{alpha}), CS3 is guaranteed.
        
        Choosing $=$ for $\geq$ in (\ref{escrow})--(\ref{alpha}), we ensure requirement CS by solving these equations.
        In particular, for $i=0,\dots,n{-}1$,
        $$a_{i} := \Phi^{n-1-i}\cdot(\Phi\cdot\epsilon+2\cdot \Delta) + \sum_{j=i+1}^{n-1} 4\cdot\Phi^{j-i-1}\cdot (\Phi\cdot\epsilon + \Delta)\;.$$
        For $i=n{-}1$ this follows by (2).
        Assume we have it for $i{+}1$.
        Then
        $$a_{i+1} = \Phi^{n-2-i}\cdot(\Phi\cdot\epsilon+2\cdot \Delta) + \sum_{j=i+2}^{n-1} 4\cdot\Phi^{j-i-2}\cdot (\Phi\cdot\epsilon + \Delta)\;.$$
        Now, applying (\ref{escrow}) and (\ref{alpha}),
        multiply by $\Phi$ and add $4\cdot\Phi\cdot\epsilon + 4\cdot\Delta$.
                
\paragraph{Liveness.}
It remains to check property L\@. Suppose that all parties abide by the protocol.  By the reasoning
in Section~\ref{initialisation} we may assume that the action $s(e_0,\$)\color{blue} \MVAt c_o$ (using the
terminology of Section~\ref{run}) of Alice sending money to her escrow will not take place until all
actions $s(c_i,G(d_i))\color{blue} \MVAt e_i$ and $r(e_i,G(d_i))\color{blue} \MVAt c_i$ have occurred.
In terms of Section~\ref{run} we show that in the remaining active phase of the protocol
at least the prefix of the displayed sequence of actions
until and including $s(c_n,\$)\color{blue} \MVAt e_{n-1}$ will take place, in that order, and not interleaved
with any other actions. When this happens, Bob must be in his
third state, and the required action $r(e_{n-1},\$)\color{blue} \MVAt c_{n}$ will follow surely.

Towards a contradiction, let the initial behaviour of the active part of the protocol be a strict prefix of this
sequence, where $a$ is the first action in the sequence that does not occur as scheduled.
A simple case analysis shows that when $a$ is scheduled, in fact no action other than $a$ is possible.

The action $a$ cannot be of the form $s(e_i,\$)\color{blue} \MVAt c_i$, because when this action is due, customer
$c_i$ is in a state where this action must be taken within time $\epsilon$. An exception is the case
$n=0$, but also here the action must be taken in within a finite amount of time (or there is nothing
to verify).

The action $a$ cannot be $r(c_i,\$)\color{blue} \MVAt e_i$ either, because each message sent must arrive
eventually, and the receiving party $e_i$ is in its second state, and thus able to perform the
receive action.

Similarly, $a$ cannot be $s(c_{i+1},P(a_i))\color{blue} \MVAt e_i$, as here the sender $e_i$ must be
in its third state.

Since all actions $r(e_i,G(d_i))\color{blue} \MVAt c_i$ have already occurred, $a$ cannot be of the form
$r(e_{i-1},P(a_{i-1}))\color{blue} \MVAt c_i$.

The case $a=s(e_{n-1},\chi)\color{blue} \MVAt c_n$ can be excluded, as here Bob must be in his second state.

If $a=r(c_n,\chi)\color{blue} \MVAt e_{n-1}$, then the recipient $e_{n-1}$ must be in its fourth
state and, due to the careful choice of $a_{n_1}$ (see (\ref{Bob})), the time-out transition can not intervene.
So that choice of $a$ is excluded as well.

The argument against $a=s(c_n,\$)\color{blue} \MVAt e_{n-1}$ is trivial.

\section{Impossibility under partial synchrony of communications}

When communications can experience arbitrarily long delays, it is not possible to expect from a protocol to terminate in an a priori known amount of time. If we want to perform cross-chain payments in a partially synchronous setting, it is therefore necessary to define a different class of cross-chain payment protocols. The obvious idea is to remove any time bound in the definition. But as we will show, this is not enough to make the problem solvable.

\begin{definition}[Eventually terminating cross-chain payment protocol with strong liveness guarantees]
  A cross-chain payment protocol is an \textit{eventually terminating cross-chain payment protocol with strong liveness guarantees} it satisfies all the properties of Definition \ref{time-bounded correct} except Property T which is replaced by the following:

  \begin{itemize}
  \item[T'] \textbf{Eventual termination.} Each customer that abides by the protocol,
    and either makes a payment or issues a certificate, terminates \emph{eventually}, provided her escrows abide by the protocol.
  \end{itemize}
\end{definition}

\begin{theorem}
  If communications are partially synchronous (assuming that customers may crash), then there is no eventually terminating cross-chain payment protocol with strong liveness guarantees.
\end{theorem}

\begin{proof}
  Assume an eventually terminating cross-chain payment protocol
  with strong liveness guarantees.

    Consider a run $r$ in which all participants abide by the protocol.  It exists by property C\@.
    By property L Bob will be paid in this run, and by property T' all customers terminate.
    Since by properties ES and CS no participant
    makes a loss, Alice will not get her money back. Hence by property CS1 Alice will end up with
    the certificate $\chi$. Let $c_i$ be the last customer that holds the certificate before it
    reaches Alice. This must be either Bob or one of the connectors.
    Let $s$ be the state in which $c_i$ is about the send on $\chi$.

    The protocol may not prescribe that $c_i$ has already received the money in state $s$. For then
    customer $c_i$ could decide to keep the money as well as the certificate, while all other
    participants keep abiding by the protocol, which would violate properties T', ES or CS\@.
  
  Now consider the following two runs of the system, that are the same until state $s$.
  In run $r_1$ customer $c_i$ never lets go of the certificate, nor sends out any other message past
  state $s$,  while all other participants abide by the protocol; in run $r_2$ all participants abide by the
  protocol, but $c_i$'s message with the certificate, and all subsequent messages from
  $c_i$, experience an extreme delay.

  First assume that $c_i$ is in fact Bob.
  By property T' run $r_1$ reaches a state $s'$ in which all customers other than Bob are terminated.
  Using properties ES and CS, in this state
  Alice ends up without the certificate, and thus with the money, and all customers Chloe$_j$ play even.
  If follows that Bob never receives his money.
  Yet for all participants other than Bob, runs $r_1$ and $r_2$ are indistinguishable, so $r_2$ will
  reach a similar state.  This violates property CS2.

  Now assume customer $c_i$ is not Bob. So Bob has already issued the certificate.
  By property T' run $r_1$ reaches a state $s'$ in which all customers other than $c_i$ are terminated.
  Using properties ES and CS, in this state
  Alice as well as Bob end up with the money, and all customers Chloe$_j$ with $j\neq i$ play even.
  If follows that Chloe$_i$ looses her money.
  Yet for all participants other than Chloe$_i$, runs $r_1$ and $r_2$ are indistinguishable, and the
  delayed certificate sent by Chloe$_i$ may arrive only after the system has reached state $s'$.
  This violates property CS3.
\end{proof}

\section{Solution to a variant under partial synchrony}

\subsection{Eventually terminating cross-chain payment protocol with weak liveness guarantees}

\paragraph{Weak liveness guarantees.}

Since it is impossible to design an \textit{eventually terminating cross-chain payment protocol with
  strong liveness guarantees} under partial synchrony, we are going to define
what is an \textit{eventually terminating cross-chain payment protocol with weak
liveness guarantees}, and show that it is possible to implement such a protocol.

In view of the impossibility proof given above, the \textit{strong liveness}
condition (L) is too strong. We would like to weaken it, and replace it by a (realistic and still
desirable) property called \textit{weak liveness} such that the problem
becomes solvable. A similar situation exists in the atomic commit problem literature, for instance
\textit{weak non-triviality} defined by Guerraoui in \cite{Gue95} or condition $AC4$ for atomic
commit defined by Hadzilacos in \cite{Had90}:\vspace{-2ex}

\begin{quote}
    If all existing failures are repaired and no new failures occur for a sufficiently long period
    of time, then all processes will reach a decision.
\end{quote}{}

\paragraph{Abort certificate.}
In the synchronous solution, we used timelocks to ensure that the money will not get stuck
forever in escrow. Under partial synchrony, it is no longer pertinent to use timelocks. We need to
replace them by a safe way to unlock the funds stored in an escrow.

We therefore modify the definition of the certificate $\chi$. Instead of having a single certificate
simply signed by Bob, we have two more general certificates called \textit{commit certificate}
$\chi_c$ and \textit{abort certificate} $\chi_a$, that can never exist simultaneously.

\paragraph{New definition of the problem.} When we take into account the two previous tweaks, we reach the following definition. We highlight in italics the difference with our previous definition of a \textit{time-bounded cross-chain payment protocol}. In particular, we replace ``Bob will not issue $\chi$" by ``Bob will receive $\chi_a$", and ``Alive will receive $\chi$'' by ``Alice will receive $\chi_c$''.
\pagebreak[4]

\begin{definition}[Eventually terminating cross-chain payment protocol with weak liveness guarantees]
\label{etccppwwlg}

  A cross-chain payment protocol is an \textit{eventually terminating cross-chain payment protocol
  with weak liveness guarantees} if it satisfies the following properties:

  \begin{itemize}
  \item[C] \textbf{Consistency.} For each participant in the protocol it is possible to abide by the protocol.
  \item[CC] \textbf{Certificate consistency.} \emph{An abort and a commit certificate can never be issued both.}
  \item[T'] \textbf{Eventual termination.} Each customer that abides by the protocol
    terminates eventually, provided her escrows abide by the protocol.
  \item[ES] \textbf{Escrow security.} Each escrow that abides by the protocol does not loose money.
  \item[CS'] \textbf{Customer security.}
    \begin{itemize}
      \item[CS1'] Upon termination, if Alice and her escrow abide by the protocol, Alice  has either got her money back or received the \textit{commit certificate} $\chi_c$.
  
      \item[CS2'] Upon termination, if Bob and his escrow abide by the protocol, Bob has either received the money \textit{or the abort certificate} $\chi_a$.
      \item[CS3'] Upon termination, each connector that abides by the protocol has got her money back,
      provided her escrows abide by the protocol. 

    \end{itemize}
    \item[L'] \textbf{Weak liveness.} If all parties abide by the protocol, \textit{and if the customers wait sufficiently long before and after sending money}, then Bob is eventually paid.
\end{itemize}

\end{definition}

\subsection{Transaction manager abstraction}\label{sec:transaction manager}
\newcommand{\TM}{\mbox{\it TM}}
\newcommand{\Agreement}{Consistency}

The previous impossibility result shows that when there is no synchrony, it is hard for processes to agree on a uniform commitment decision (abort or commit). We are going to leverage the existing solutions to the classical consensus problem, and embed them in an abstraction called ``transaction manager'', defined as follows:

\begin{definition}[Transaction manager]
  \label{transaction_manager_def}

    A \textit{transaction manager}, called $\TM$, is a process that can receive binary values from the set $\{\textsc{commit}, \textsc{abort}\}$ from a customer, and that can send certificates drawn from $\{\chi_c, \chi_a\}$ to customers.
    
    A transaction manager must verify the following properties:
  
    \begin{itemize}
        \item \textbf{TM-Termination}. If a customer proposes to \textsc{abort}, or if Bob proposes
          to \textsc{commit}, then $\TM$ sends eventually a certificate $\chi_c$ or
          $\chi_a$ to every customer.
          If a customer proposes \textsc{abort} or \textsc{commit} after $\TM$
          has issued a certificate, $\TM$ will send a copy of that certificate to that customer.
        \item \textbf{TM-\Agreement}. $\TM$ does not issue two different certificates.
        \item \textbf{TM-Commit-Validity.} $\TM$ can issue $\chi_c$ only if Bob proposed \textsc{commit}.
        \item \textbf{TM-Abort-Validity.} $\TM$ can issue $\chi_a$ only if some customer proposed \textsc{abort}.
    \end{itemize}
\end{definition}

~\\There are several ways of implementing a transaction manager.

\begin{itemize}
    \item \textbf{Centralised transaction manager.} The transaction manager can be a centralised actor trusted by every customer.
    \item \textbf{Distributed transaction manager.} The transaction manager could be a collection of
      $k$ parties (``validators'') appointed by the participants in the protocol.
      These validators could run the consensus algorithm for partial synchrony from Dwork, Lynch \&
      Stockmeyer \cite{DLS88}, or any equivalent algorithm. This works when less than one third of
      these validators are unreliable. In this case, ``sending a message to the TM'' means sending
      it to each of these validators, and ``the TM sending a decision'' means strictly more than one third of the validators
      sending the jointly taken decision.
    
    \item \textbf{External decentralised transaction manager.} The transaction manager can be a
      decentralised data structure. For
      example, a smart contract running on a permissionless blockchain shared by every
      customer can be programmed to be a transaction manager.
\end{itemize}

In the following, we assume that we have such a transaction manager. Even if it is run by the
customers, we do not specify the messages exchanged to run it: the transaction manager is a
black-box embedding a consensus algorithm that is running off-protocol.
Section \ref{implementation-TM} provides an example implementation.

\subsection{A protocol responding to the problem}
\vspace{2ex}
\vfill

\begin{center}
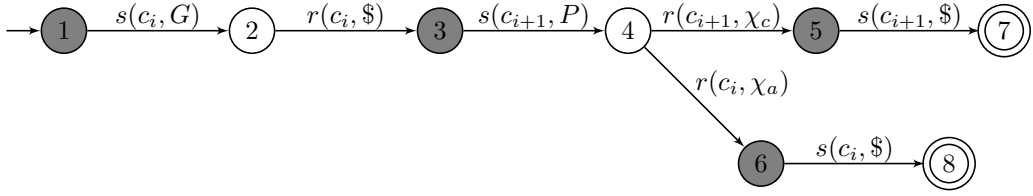

    \begin{tikzpicture}[>=latex', auto, semithick, node distance=2.5cm, initial text={}, double distance=2pt, every edge/.style={draw,->}, every state/.style={inner sep=0cm, minimum size=0.6cm}]

        \node[state, initial, fill=gray]   (A)                     {1};
        \node[state]            (B) [right of=A]        {2};
        \node[state, fill=gray] (C) [right of=B]        {3};
        \node[state]            (D) [right of=C]        {4};
        \node[state, fill=gray] (E) [right of=D]  {5};
        \node[state, fill=gray] (F) [below right of=D]  {6};
        \node[state, accepting] (G) [right of=E]        {7};
        \node[state, accepting]            (H) [right of=F]        {8};

        \path   (A) edge        node {$s(c_i, G)$}          (B)
                (B) edge        node {$r(c_i, \$)$}      (C)
                (C) edge        node {$s(c_{i+1}, P)$}         (D)
                (D) edge        node {$r(c_{i+1}, \chi_c)$}     (E)        
                (D) edge        node {$r(c_i, \chi_a)$}     (F)
                (E) edge        node {$s(c_{i+1}, \$)$}          (G)        
                (F) edge        node {$s(c_i, \$)$}          (H)            ;

      \end{tikzpicture}
      \captionof{figure}{Automaton for escrow $e_i$}

    \end{center}
\vfill

\begin{center}
    \begin{tikzpicture}[>=latex', auto, semithick, node distance=2.5cm, initial text={}, every edge/.style={draw,->}, double distance=2pt, every state/.style={inner sep=0cm, minimum size=0.6cm}]

        \node[state, initial]   (A)                     {1};
        \node[state]            (B) [right of=A]        {2};
        \node[state, accepting] (M) [below right of=B]  {3};
        \node[state, fill=gray] (C) [right of=B]        {4};
        \node[state]            (D) [right of=C]        {5};
        \node[state, fill=gray] (E) [above right of=D]  {6};
        \node[state, fill=gray] (F) [below right of=D]  {7};
        \node[state, fill=gray] (G) [right of=D]        {8};
        \node[state]            (H) [right of=G]        {9};
        \node[state]            (I) [right of=E]        {10};
        \node[state, accepting] (J) [right of=I]        {11};
        \node[state]            (K) [right of=F]        {12};
        \node[state, accepting] (L) [right of=K]        {13};

        \path   (A) edge        node {$r(e_i, G)$}          (B)
                (B) edge        node {$r(e_{i-1}, P)$}      (C)
                (C) edge        node {$s(e_i, \$)$}         (D)
                (D) edge        node {$r(\TM, \chi_c)$}     (E)        
                (D) edge        node {$r(\TM, \chi_a)$}     (F)
                (D) edge        node {$now \geq T_i$}          (G)        
                (G) edge        node {$s(\TM, \textsc{ab.})$}          (H)        
                (H) edge[above] node {$\qquad\;\;\quad r(\TM, \chi_c)$}     (E)        
                (H) edge[above] node {$\hspace{80pt} r(\TM, \chi_a)$}     (F)
                (E) edge        node {$s(e_{i-1}, \chi_c)$} (I)    
                (I) edge        node {$r(e_{i-1}, \$)$}     (J)        
                (F) edge        node {$s(e_i, \chi_a)$}     (K)        
                (K) edge        node {$r(e_i, \$)$}         (L)        
                (B) edge        node {$now \geq T_i$}          (M)        

        ;

      \end{tikzpicture}
    \end{center}
\begingroup

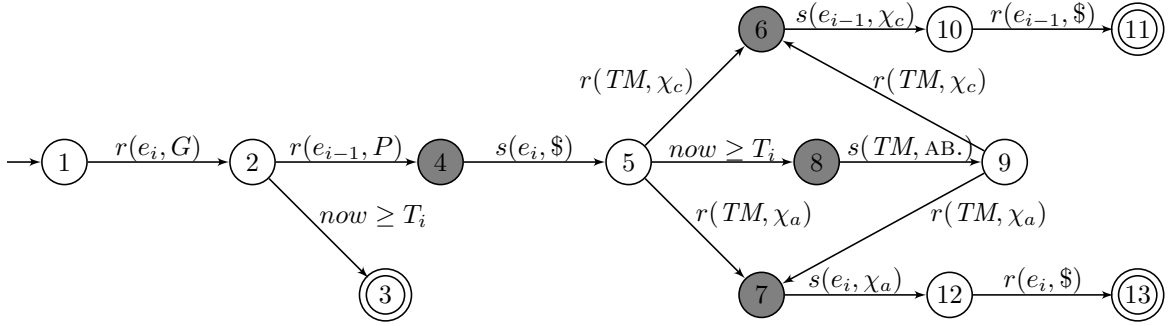
\captionof{figure}{Automaton for customer $c_i, i \in \{1..n{-}1\}$}
\endgroup
\vfill

\begin{center}
    \begin{tikzpicture}[>=latex', auto, semithick, node distance=2.5cm, initial text={}, every edge/.style={draw,->}, double distance=2pt, every state/.style={inner sep=0cm, minimum size=0.6cm}]

        \node[state, initial]   (A)                     {1};
        \node[state, fill]      (C) [right of=A]        {4};
        \node[state]            (D) [right of=C]        {5};
        \node[state, accepting] (E) [above right of=D]  {6};
        \node[state, fill=gray] (F) [below right of=D]  {7};
        \node[state, fill=gray] (G) [right of=D]        {8};
        \node[state]            (H) [right of=G]        {9};
        \node[state]            (K) [right of=F]        {12};
        \node[state, accepting] (L) [right of=K]        {13};

        \path   (A) edge        node {$r(e_0, G)$}          (C)
                (C) edge        node {$s(e_0, \$)$}         (D)
                (D) edge        node {$r(\TM, \chi_c)$}     (E)        
                (D) edge        node {$r(\TM, \chi_a)$}     (F)
                (D) edge        node {$now \geq T_0$}          (G)        
                (G) edge        node {$s(\TM, \textsc{ab.})$}          (H)        
                (H) edge[above] node {$\qquad\;\;\quad r(\TM, \chi_c)$}     (E)        
                (H) edge[above] node {$\hspace{80pt} r(\TM, \chi_a)$}     (F)
                (F) edge        node {$s(e_0, \chi_a)$}     (K)        
                (K) edge        node {$r(e_0, \$)$}         (L)        
        ;
      \end{tikzpicture}
    \end{center}
\begingroup

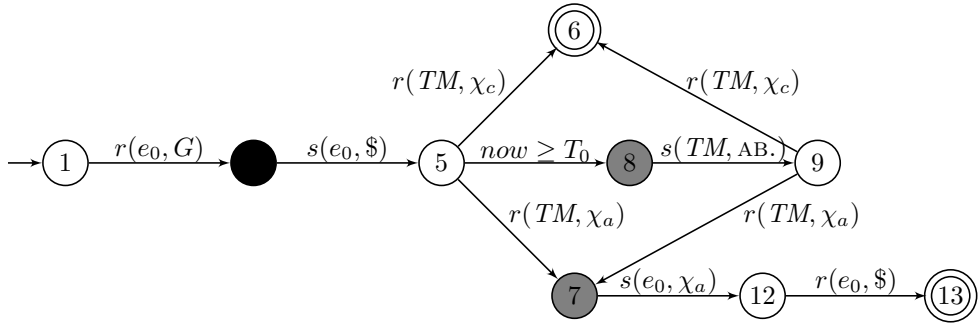
\captionof{figure}{Automaton for customer $c_0$ (Alice)}
\endgroup

\begin{center}
    \begin{tikzpicture}[>=latex', auto, semithick, node distance=2.5cm, initial text={}, every edge/.style={draw,->}, double distance=2pt, every state/.style={inner sep=0cm, minimum size=0.6cm}]

        \node[state, initial]   (A)                     {1};
        \node[state, fill=gray] (C) [right of=A]        {3};

        \node[state]            (D) [right of=C]        {5};
        \node[state, fill=gray] (E) [above right of=D]  {6};
        \node[state, accepting] (F) [below right of=D]  {7};
        \node[state]            (I) [right of=E]        {10};
        \node[state, accepting] (J) [right of=I]        {11};
        \node[state, fill=gray] (M) [below right of=A]  {2};
        \node[state]            (N) [right of=M]        {4};

        \path   (A) edge        node {$r(e_{n-1}, P)$}          (C)
                (C) edge        node {$s(\TM, \textsc{com.})$}         (D)
                (D) edge        node {$r(\TM, \chi_c)$}     (E)        
                (D) edge        node {$r(\TM, \chi_a)$}     (F)
                (E) edge        node {$s(e_{n-1}, \chi_c)$} (I)    
                (I) edge        node {$r(e_{n-1}, \$)$}     (J)
                (A) edge        node {$now \geq T_n$}       (M)        
                (M) edge        node {$s(\TM, \textsc{ab.})$}         (N)
                (N) edge        node {$r(\TM, \chi_a)$}     (F)
  ;      
      \end{tikzpicture}
    \end{center}
\begingroup

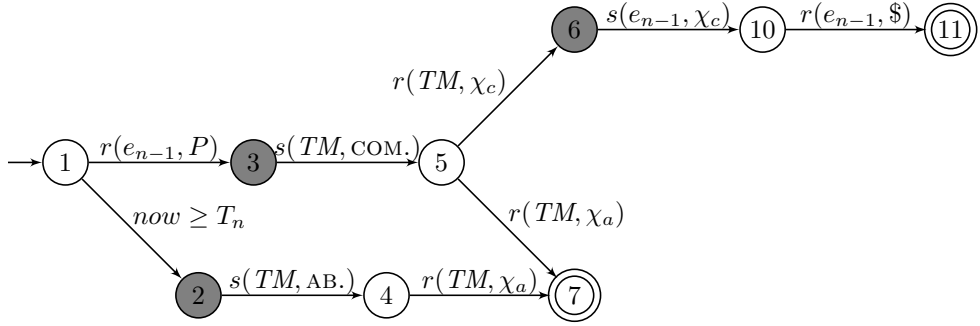
\captionof{figure}{Automaton for customer $c_n$ (Bob)}
\endgroup

\paragraph{Description of the protocol.} In this version of the protocol Chloe$_i$ awaits the two promises of her escrows, just like in
the protocol of Figure~\ref{automata}, and then sends the money to her downstream escrow.
Subsequently she awaits an abort or commit certificate from the transaction manager.
If she gets a commit certificate, she cashes it in at her upstream escrow to obtain the money.
If she gets an abort certificate instead, she cashes it in at her downstream escrow for a refund of
the money she paid earlier. In case she looses patience before she gets the second promise, which
happens at a time $T_i$ specific for Chloe$_i$, she simply quits. In case she looses patience after
she has invested the money but before she gets any certificate, the timeout transition occurs, and
she sends an abort proposal to the transaction manager. The latter replies on this with either an
abort or a commit certificate, and she cashes those in as above.

The tags $G$ and $P$ can thus be understood as promises that would tell:

\begin{itemize}
    \item $G$: ``I guarantee that if I receive \$ from you, then if you send me $\chi_a$ I will send you \$".
    \item $P$. ``I promise that if you send me $\chi_c$ I will send \$ to you".
\end{itemize}
  The automaton for Alice is just a simplified version of the one for Chloe, and the one for the
  escrows is trivial.
  For Bob, the important modification is that he alerts the TM with a {\sc commit} message when the
  protocol is ready for this. Moreover, in case Bob looses patience before receiving any promise, he
  send a {\sc abort} message to the TM, so that he receives the abort certificate in response.

\subsection{Proof of correctness}

Let us call $P$ the protocol defined by the above ANTA.
Section \ref{initialisation} (Initialisation) applies to $P$ as well, and we assume that the
  appropriate modifications are made.

\begin{theorem}
  Protocol $P$ is an eventually terminating cross-chain payment protocol with weak liveness guarantees.
\end{theorem}

\begin{proof}
  The properties to prove are close to the time-bounded cross-chain payment protocol, and the proof is similar. We split the proof in the following lemmas: Lemma \ref{consistency2}, \ref{cc}, \ref{termination2}, \ref{ls2}, \ref{cs2} and \ref{progress2}.
\end{proof}

\begin{lemma}[Consistency]
  \label{consistency2}
  For each participant in $P$ it is possible to abide by $P$.
\end{lemma}

\begin{proof}
  We have to ensure that each participant will be able to follow the transitions after a grey or black state. The only way this would not be possible would be when the protocol asks to transmit a resource that is not available. It is always possible to send tags and money, but we need to verify that any sending of a certificate is preceded by the receipt of this certificate---except for the issuer of the certificate. Such a property is clear after inspection of the automata of escrows and customers.
\end{proof}

\begin{lemma}[Certificate consistency]
  \label{cc}
An abort and a commit certificate can never be issued both.
\end{lemma}
\begin{proof}
This is a consequence of Property \textbf{TM-Consistency} of Definition~\ref{transaction_manager_def}.
\end{proof}

\begin{lemma}[Eventual termination]
  \label{termination2}
  Each customer that abides by $P$ will terminate eventually, provided her escrows abide by the protocol. 
\end{lemma}

\begin{proof} Thanks to Lemma \ref{consistency2}, in order to show that a terminating
  state will be reached, it is sufficient to prove that every input state will be left eventually.
  We thus establish Lemma~\ref{termination2} by showing that Chloe$_i$ and Alice cannot get
  stuck in states 1, 2, 5, 9, 10 and 12, and Bob cannot get stuck in states 1, 4, 5 and 10.

  Since her downstream escrow will surely leave state 1, and thus send the message $G$, it follows that
  Chloe$_i$, and Alice, will receive this message, and thereby leave state 1.

  Chloe$_i$ will leave state $2$ at time $T_i$ at the latest, or immediately when
  reaching this state after time $T_i$. By the same reasoning, Bob will not get stuck in state 1.

  Chloe$_i$ and Alice will leave state $5$ at time $T_i$ at the latest, or immediately when
    reaching this state after time $T_i$. She will send an {\sc abort} message to $\TM$ to
    reach state $9$, and by the property \textbf{TM-Termination}
    of Definition~\ref{transaction_manager_def}, $\TM$ will eventually reply to her with $\chi_a$ or $\chi_a$.
    Hence she will leave state 9.
    By the same reasoning, Bob will not get stuck in state 5.

  To reach state 4, Bob sends an {\sc abort} message to $\TM$. By properties \textbf{TM-Termination}
    and \textbf{TM-Commit-Validity} of Definition~\ref{transaction_manager_def},  $\TM$ will
    eventually reply to him with $\chi_a$. Hence he will leave state~4.

      Now assume Customer $c_{i+1}$ (Chloe or Bob) reaches state 10. Then Customer $c_{i+1}$ has
      already received promise $P$ from Escrow $i$, and thus Escrow $i$ must have send this promise,
      thereby reaching state 4. To reach state 10,  Customer $c_{i+1}$ sends certificate $\chi_c$ to
      Escrow $i$. By Lemma~\ref{cc}, the TM never issues certificate $\chi_a$, so Customer $c_i$
      cannot send it to Escrow $i$. It follows that Escrow $i$ will reach state 5,
      and send the money to Customer $c_{i+1}$. Hence Customer $c_{i+1}$ will leave state 10
      and reach the terminating state $11$.

      Finally assume Customer $c_{i}$ (Alice or Chloe) reaches state 12.  Then Customer $c_{i}$ has
      already received promise $G$ from Escrow $i$, and thus Escrow $i$ must have send this promise,
      thereby reaching state 2. After receiving promise $G$, Customer $c_{i}$ has send the money to
      Escrow $i$, so Escrow $i$ will have reached state 3, and hence also state 4.
      To reach state 12,  Customer $c_{i}$ sends certificate $\chi_a$ to
      Escrow $i$. By Lemma~\ref{cc}, the TM never issues certificate $\chi_c$, so Customer $c_{i+1}$
      cannot send it to Escrow $i$. It follows that Escrow $i$ will reach state 6,
      and send the money to Customer $c_{i}$. Hence Customer $c_{i}$ will leave state 12
      and reach the terminating state $13$.
\end{proof}

\begin{lemma}[Escrow-security]
  \label{ls2}
  Each escrow that abides by $P$ will not loose money.
\end{lemma}
\begin{proof}
  The result is immediate: any transition where an escrow sends money has been preceded by a
  transition where it receives the money. An escrow's balance cannot become negative if it follows the protocol.
\end{proof}

\begin{lemma}[Customer-security]
  \label{cs2}
  ~    \vspace{-3pt}
  \begin{enumerate}
    \item Upon termination, if Alice and her escrow abide by $P$, Alice has either got her
      money back or received the commit certificate $\chi_c$.
    \vspace{-5pt}
    \item Upon termination, if Bob and his escrow abide by $P$, Bob has either received the
      money or the abort certificate $\chi_a$.
    \vspace{-5pt}
    \item Upon termination, each connector that abides by $P$ has got her money back,
      provided her escrows abide by $P$.
  \end{enumerate}
\end{lemma}

\begin{proof}
  ~  \vspace{-5pt}
  \begin{enumerate}
  \item
    If Alice terminates in state $13$, she has got her money back in the last
    transition. If she terminates in state $6$, she has got the certificate $\chi_c$ in the
    last transition.
    \vspace{-5pt}
    \item The result is similar for Bob: if he terminates in state $7$, he has received
      $\chi_a$. Otherwise, he terminates in state $11$ and has been paid correctly.
    \vspace{-5pt}
  \item
    To reach termination, Chloe$_i$, $i \in \{1..n-1\}$ has either never spend the money
    (termination state $3$), or received $\$$ either from 
    $e_{i-1}$ (termination state $11$) or from $e_i$ (termination state $13$). In both these cases, she
    has got her money back.\qedhere
  \end{enumerate}
\end{proof}

\begin{lemma}[Weak liveness]
  \label{progress2}
  If all participants abide by $P$, and if the customers wait sufficiently long before and after sending money, then Bob will be paid.
\end{lemma}

\begin{proof}
  Let us suppose that all the participants abide by $P$.
  Suppose that for all $i \in \llbracket0,n{-}1\rrbracket$, $T_i$ is large enough for Alice, Bob and every connector
  to never take any
    time-out transition. Instead, Bob will be the first customer to call the
  transaction manager in his transition from state 3 to 5.

  Using the notation of Section~\ref{run}, the active part of the protocol---after the exchange of
    the tags $G$---must start with the following sequence of actions, executed in this order:
  \begin{center}
  \begin{tabular}{@{}ll@{}crrr@{}}
        $s(e_0,\$) \color{blue} \MVAt c_0$ & $r(c_0,\$)\color{blue}\MVAt  e_0$ &&
        $s(c_1,P) \color{blue}\MVAt e_0$ & $r(e_0,P) \color{blue} \MVAt c_1$ \\
        $s(e_1,\$) \color{blue} \MVAt c_1$ & $r(c_1,\$)\color{blue}\MVAt  e_1$ &&
        $s(c_2,P) \color{blue}\MVAt e_1$ & $r(e_1,P) \color{blue} \MVAt c_2$ \\
        \dots \\
        $s(e_{n-1},\$) \color{blue} \MVAt c_{n-1}$ & $r(c_{n-1},\$)\color{blue}\MVAt  e_{n-1}$ &&
        $s(c_n,P) \color{blue}\MVAt e_{n-1}$ & $r(e_{n-1},P) \color{blue} \MVAt c_n$ &
        $s(\TM,\textsc{com.}) \color{blue} \MVAt c_n$
  \end{tabular}
  \end{center}
  By the \textbf{TM-Abort-Validity} property of Definition \ref{transaction_manager_def}, since the only proposal was \textsc{commit}, $\TM$ will issue the certificate $\chi_c$. In particular, Bob will give it to $e_{n-1}$ and receive the payment in exchange.
\end{proof}

  Interestingly, nothing in the above proofs depends in any way on the assumption of partially
  synchronous communication. The only place where this is needed is for the implementation of the
  transaction manager $\TM$---see Section~\ref{implementation-TM}. In case one is content with a
  centralised transaction manager as described in Section \ref{sec:transaction manager}, our
  protocol works correctly also when assuming communication to be asynchronous.

\subsection{Implementation of a decentralised transaction manager}
\label{implementation-TM}

As an example, in this section we provide an explicit implementation of a transaction manager.
It is a wrapper, expressed in pseudo-code, around a binary Byzantine consensus algorithm, such as
the one from~\cite{DLS88}, which is treated as a black box.
It follows that we do not need a shared blockchain to achieve our cross-chain payment protocol.

Suppose that we have $m$ \emph{validators}, which are agents running a consensus algorithm.
The validators communicate which each other by exchanging messages.
We suppose that a certain number $f$ of validators can be \emph{faulty},
in the sense that we allow arbitrary (Byzantine) behaviour. All other validators are assumed to abide by the
protocol defining a consensus algorithm.

\begin{definition}[Binary Byzantine Consensus]
\label{consensus_correctness} A binary Byzantine consensus algorithm is an algorithm in which
every validator can \emph{propose} a binary value (\textit{i.e.}\ in $\{0,1\}$) and \emph{decide} a binary value.
Assuming that every non-faulty validator proposes a binary value, the following properties must be ensured:
\begin{itemize}
    \item \textbf{BBC-Termination.} Every non-faulty process eventually decides on a binary value.
    \item \textbf{BBC-Agreement.} No two non-faulty processes decide on different binary values.
    \item \textbf{BBC-Validity.} If all non-faulty processes propose the same value, no other value can be decided.
\end{itemize}
\end{definition}

\begin{theorem}[\cite{DLS88}]
Assuming partially synchronous communication, a binary Byzantine consensus algorithm exists when $f < m/3$.
\end{theorem}

\noindent
Using such an algorithm as a black box, we now implement a $\TM$.
Our validators can either be customers, like Alice, Chloe and Bob, or external parties.
If the set of validators is included in or equal to the set of customers, we can talk of an \emph{internal
decentralised transaction manager}.
Our TM implementation is only valid when assuming partially synchronous communication, and $f < m/3$.
Each customer can communicate with every validator.

\paragraph{Reliable broadcast call.} To call the transaction manager, a customer reliably broadcasts
a message to all validators. Here a \emph{reliable broadcast} is a protocol described by Bracha in \cite{Bra87}.
It is guaranteed to terminate, even in a setting with asynchronous communication, provided less than
one-third of all the broadcast recipients is faulty --- the rest abiding by the protocol.
It guarantees that if the sender abides by the protocol, all recipients will receive the message sent.
Moreover, even if the sender is faulty, either all correct recipients agree on the same value sent,
or none of them accepts any value as having been sent \cite{Bra87}.

If a faulty customer sends a call with different values to different
validators, or sends something to some validators and nothing to others, the reliable broadcast
primitive will filter these messages out. In particular, if a validator receives a call from a
customer then eventually every validator will receive this call from this customer.

\paragraph{Certificate implementation.} With $\sigma_k(v)$ we denote the value $v\in\{0,1\}$
cryptographically signed by validator $k$. Such a signed message models the decision {\sc abort} (if $v=0$)
or {\sc commit} ($v=1$) taken by validator $k$. Since up to $f$ validators may be unreliable,
a valid certificate is a message that contains (for instance as attachments) more than $f$ 
copies of the same decision, taken by different validators $k$. We model such certificates as sets.
Hence:
\begin{itemize}
  \item $\chi_c$ is any set of at least $f{+}1$ messages $\sigma_k(1)$ signed by at least $f{+}1$ different validators $k$.
  \item $\chi_a$ is any set of at least $f{+}1$ messages $\sigma_k(0)$ signed by at least $f{+}1$ different validators $k$.
\end{itemize}
Such a certificate is verifiable non-interactively by a third party such as any customer or
escrow. Asking for $f{+}1$ signatures ensures that at least one correct validator has issued this
certificate, and in particular will guarantee property CC (certificate consistency). Of course a
real implementation will not rely on simple signatures of the string "0" or "1" because of the
possibility of replay attacks.

\paragraph{TM implementation.} When our protocol prescribes $s(\TM,\textsc{ab.})$, resp.\ $s(\TM,\textsc{com.})$, this is implemented as
reliably broadcasting $\textsc{abort}$, resp.\ \textsc{commit}, to all validators. The transition $r(\TM,\chi)$
denotes the receipt of certificate $\chi$ from one of the validators.
The behaviour of $\TM$ is described as Algorithm~\ref{bft-tm}.

\begin{figure}
\vspace{-2ex}
\begin{algorithm}[H]
\caption{\textit{BFT-TM} algorithm for validator $v_k$, $k\in [1..m]$.}
\label{bft-tm}

\begin{algorithmic}[1]
  \Statex\\
When validator $k$ has not proposed, nor decided, a value so far:\\
\hskip\algorithmicindent when receive \textsc{abort} from a customer $c_i$, $i \in [0.. n]$:\\
\hskip\algorithmicindent\hskip\algorithmicindent propose(0) to BBC\\
\hskip\algorithmicindent when receive \textsc{commit} from $c_n$:\\
\hskip\algorithmicindent\hskip\algorithmicindent propose(1) to BBC
\Statex\\
When validator $k$ decides value $v$ (by running BBC):\\
\hskip\algorithmicindent broadcast($\sigma_i(v)$) to all validators\\
\hskip\algorithmicindent await receipt of $\sigma_j(v)$ for all $j$ in a certain $J \subseteq [1..m]$
   such that $|J| > f$\\
\hskip\algorithmicindent $\chi := \{\sigma_{j}(v) ,\; j \in J\}$ \\
\hskip\algorithmicindent broadcast($\chi$) to all customers
\Statex\\
When validator $k$ has decided a value $v$:\\
\hskip\algorithmicindent when receive \textsc{abort} or \textsc{commit} from a customer $c_i$, $i \in [0..n]$:\\
\hskip\algorithmicindent\hskip\algorithmicindent send $\chi$ to customer $c_i$.
\end{algorithmic}
\end{algorithm}
\end{figure}

\begin{theorem}[BFT-TM correctness]
  The BFT-TM algorithm implements a $\TM$ as defined in Definition \ref{transaction_manager_def}.
\end{theorem}

\begin{proof}
The correctness properties of a transaction manager derive almost immediately from the correctness
properties of a binary Byzantine consensus algorithm.

\begin{enumerate}
\item \textbf{TM-Termination.} Let us suppose that a customer sends an abort proposal to $\TM$, i.e.\ to each and every validator,
or that Bob proposes to commit. Then every correct validator starts participating in BBC with
the initially proposed binary value $0$ or $1$, respectively.
By \textbf{BBC-Termination}, every correct validator eventually passes line $6$ of the BFT-TM algorithm.
By assumption on the number of correct validators, every validator eventually
receives at least $f{+}1$ signatures on this value, thereby forming a certificate
that is sent to every customer.\pagebreak[2]

If a customer proposes \textsc{abort} or \textsc{commit} after $\TM$ has issued a certificate,
  each correct validator will send a copy of that certificate to that customer by lines 11--13 of the algorithm.
\item \textbf{TM-\Agreement.} By contradiction, let us suppose that two customers receive different
  certificates. Since a certificate contains at least $f{+}1$ signatures, a correct validator has
  broadcast $\sigma_i(0)$ and another correct validator has broadcast
  $\sigma_j(1)$. Line $6$ of BFT-TM shows that the value signed and broadcast has been decided by
  BBC\@. This is a contradiction with \textbf{BBC-Agreement}.
\item \textbf{TM-Commit-Validity.} TM-Commit-Validity says that if Bob does not propose commit, then $\TM$ cannot issue $\chi_c$. If Bob does not
  propose commit, then no correct validator will ever propose $1$ because of lines $1{-}5$ of BFT-TM.
  Now \textbf{BBC-Validity} implies that no correct validator can decide $1$, and consequently
    the commit certificate cannot be issued by $\TM$.
\item \textbf{TM-Abort-Validity.} TM-Abort-Validity says that if no customer proposes to abort, then
  $\TM$ cannot issue $\chi_a$. If no customer proposes to abort, then no correct validator will ever
  propose $0$ because of lines $1-5$ of BFT-TM. The final argument is the same as above.
  \qedhere
\end{enumerate}
\end{proof}

\section{Cross-chain deals versus cross-chain payments}\label{deals}

In a preprint by Herlihy, Liskov and Shrira \cite{Her19}, a \emph{cross-chain deal} is given by a
matrix $M$ where $M_{i,j}$ is listing an asset to be transferred from party $i$ to party $j$.
A cross-chain deal can also be represented as a directed graph, where each vertex represents a
party, and each arc a transfer; there is an arc from $i$ to $j$ labelled $v$ iff $M_{i,j} = v$ and $v\neq 0$.

They present two solutions to the problem of implementing such a deal, while aiming to ensure:
\begin{itemize}
\item \textbf{Safety.}  For every protocol execution, every compliant party ends up with an acceptable payoff.
\item \textbf{Weak liveness.}  No asset belonging to a compliant party is escrowed forever.
\item \textbf{Strong liveness.}  If all parties are compliant and willing to accept their proposed payoffs, then all transfers happen.
\end{itemize}
Here a payoff is deemed \emph{acceptable} to a party $i$ in the deal if party $i$ either receives
all assets $M_{j,i}$ while parting with all assets $M_{i,j}$, or if party $i$ looses nothing at all;
moreover, any outcome where she looses less and/or gains more then an acceptable outcome is also acceptable.

Each entry $M_{i,j}$ contains a type of asset and a magnitude---for instance ``5 bitcoins''.
For each type of asset a separate blockchain is assumed that acts as escrow.
The programming of these blockchains is assumed to be open source, so that all parties can convince
themselves that all escrows abide by the protocol. Unlike in our problem statements, this condition
is taken for granted and not stated explicitly in the problem description.
With this in mind, \textbf{Weak liveness} corresponds with our \textbf{Eventual termination} of
Definition~\ref{etccppwwlg}, while \textbf{Safety} is the counterpart of our \textbf{Customer security}.
Our requirement of \textbf{Escrow security} is left implicit in \cite{Her19}; since blockchains
do not possess any assets to start with, they surely cannot loose them. Finally, \textbf{Strong liveness}
is the counterpart of our \textbf{Strong liveness} property of Definition~\ref{time-bounded correct}.

Herlihy, Liskov and Shrira \cite{Her19} offer a timelock commit protocol that requires synchrony,
and assures all three of the above correctness properties. They also offer a certified blockchain
commit protocol that requires partial synchrony and a certified blockchain, and ensures 
\textbf{Safety} and \textbf{Weak liveness}; no protocol can offer \textbf{Strong liveness}
in a partially synchronous environment.
For both protocols the correctness is proven for so-called \emph{well-formed} cross-chain deals:
those whose associated directed graph is strongly connected.

The cross-chain payment cannot be seen as a special kind of cross-chain deal.
In first approximation, a cross-chain payment looks like a non-well-formed deal of the form
$$
\begin{bmatrix}
  & 0 & \$ &  & & &\\ 
  &  & 0 & \$  & & (0) &\\ 
  &  &   & 0 & \ddots & &\\ 
  &  &     &  & \ddots & \$ &\\ 
  &  & (0) &  & & 0 & \$ &\\ 
  &  &  &  & & & 0
 \end{bmatrix}
 \qquad \cong  \qquad
 c_0 \xrightarrow{\$} c_1 \xrightarrow{\$} \dots \xrightarrow{\$} c_n\;.
$$ 
However, this representation abstracts from the certificate $\chi$ that plays an essential role in
the statement of the time-bounded cross-chain payment problem. Factoring in $\chi$, an alternative
representation would be
$$
\begin{bmatrix}
  & 0 & \$ &  & & &\\ 
  &  & 0 & \$  & & (0) &\\ 
  &  &   & 0 & \ddots & &\\ 
  &  &     &  & \ddots & \$ &\\ 
  &  & (0) &  & & 0 & \$ &\\ 
  &\chi &  &  & & & 0
 \end{bmatrix}
\qquad \mbox{or} \qquad
\begin{bmatrix}
  & 0 & \$ &  & & &\\ 
  & \chi & 0 & \$  & & (0) &\\ 
  &  & \chi & 0 & \ddots & &\\ 
  &  &     &  & \ddots & \$ &\\ 
  &  & (0) &  & \chi & 0 & \$ &\\ 
  &  &  &  & & \chi & 0
 \end{bmatrix}
\;.
$$
However, these solutions presume a shared blockchain between Alice and Bob for the transfer of the certificate;
this runs counter to the problem description of cross-chain payments.

Conversely, neither is there a reduction from the cross-chain deal problem to the cross-chain payment problem.
A deal presented by a cyclic graph can be represented as a cross-chain payment where Alice and Bob
are identified. For instance, an atomic swap between two customers A and C can be expressed as a cross-chain
payment with three customers:
$$
\begin{bmatrix}
  & 0 & a \\ 
  & b & 0 &\\ 
 \end{bmatrix}
\qquad \cong \qquad
A \xrightarrow{a} C \xrightarrow{b} B = A \;.
$$
However, this idea does not generalise to well-formed cross-chain deals in general.
Since every strongly connected graph can be represented as a single cycle with repeated elements,
there is an obvious candidate reduction of such deals to cross-chain payments, simply by identifying
suitable intermediaries Chloe$_i$ and Chloe$_j$. However, this reduction does not preserve the
safety property of cross-chain deals; for when the deal goes through for Chloe$_j$ but is
aborted for Chloe$_i$, the resulting outcome is not (necessarily) acceptable for the unified
participant Chloe$_{\{i.j\}}$.

\section{Conclusion}

We formalised the problem of cross-chain payment with success guarantees. We show that there is no
solution to the existing variant of this problem without assuming synchrony, and offer such a
solution---one that works even in the presence of clock skew. We then relax the liveness guarantee
of this problem in order to propose a solution that works in a partially synchronous setting. 
This new problem differs from existing ones in that it prevents all participants from always
aborting, hence guaranteeing success when possible. Besides the new problem statements and our impossibility
result, an interesting aspect of our work is to relate recent blockchain problems, like interledger payments,
to the classic problem of transaction commit, and to offer Byzantine fault tolerant solutions to these.

\subsection*{Acknowledgements}
This research is supported under Australian Research Council Discovery Projects funding scheme (project number 180104030) entitled ``Taipan: A Blockchain with Democratic Consensus and Validated Contracts'' and Australian Research Council Future Fellowship funding scheme (project number 180100496) entitled ``The Red Belly Blockchain: A Scalable Blockchain for Internet of Things''.

\bibliographystyle{eptcs}
\bibliography{references}

\end{document}

%% file: automata.tex
\expandafter\ifx\csname graph\endcsname\relax
   \csname newbox\expandafter\endcsname\csname graph\endcsname
\fi
\ifx\graphtemp\undefined
  \csname newdimen\endcsname\graphtemp
\fi
\expandafter\setbox\csname graph\endcsname
 =\vtop{\vskip 0pt\hbox{%
    \graphtemp=.5ex
    \advance\graphtemp by 0.000in
    \rlap{\kern 0.640in\lower\graphtemp\hbox to 0pt{\hss $e_i$: Escrow $i$ ($i=0,\dots,n{-}1$)\hss}}%
\pdfliteral{
q [] 0 d 1 J 1 j
0.576 w
0.576 w
33.48 -25.128 m
33.48 -29.740682 29.740682 -33.48 25.128 -33.48 c
20.515318 -33.48 16.776 -29.740682 16.776 -25.128 c
16.776 -20.515318 20.515318 -16.776 25.128 -16.776 c
29.740682 -16.776 33.48 -20.515318 33.48 -25.128 c
h q 0.5 g
B Q
0.072 w
q 0 g
9.576 -23.328 m
16.776 -25.128 l
9.576 -26.928 l
9.576 -23.328 l
B Q
0.576 w
0 -25.128 m
9.576 -25.128 l
S
117.216 -25.128 m
117.216 -29.740682 113.476682 -33.48 108.864 -33.48 c
104.251318 -33.48 100.512 -29.740682 100.512 -25.128 c
100.512 -20.515318 104.251318 -16.776 108.864 -16.776 c
113.476682 -16.776 117.216 -20.515318 117.216 -25.128 c
S
0.072 w
q 0 g
93.24 -23.328 m
100.44 -25.128 l
93.24 -26.928 l
93.24 -23.328 l
B Q
0.576 w
33.48 -25.128 m
93.24 -25.128 l
S
Q
}%
    \graphtemp=\baselineskip
    \multiply\graphtemp by -1
    \divide\graphtemp by 2
    \advance\graphtemp by .5ex
    \advance\graphtemp by 0.349in
    \rlap{\kern 0.930in\lower\graphtemp\hbox to 0pt{\hss $s(c_i,G(d_i))\;$\hss}}%
\pdfliteral{
q [] 0 d 1 J 1 j
0.576 w
200.88 -25.128 m
200.88 -29.740682 197.140682 -33.48 192.528 -33.48 c
187.915318 -33.48 184.176 -29.740682 184.176 -25.128 c
184.176 -20.515318 187.915318 -16.776 192.528 -16.776 c
197.140682 -16.776 200.88 -20.515318 200.88 -25.128 c
h q 0.5 g
B Q
0.072 w
q 0 g
176.976 -23.328 m
184.176 -25.128 l
176.976 -26.928 l
176.976 -23.328 l
B Q
0.576 w
117.216 -25.128 m
176.976 -25.128 l
S
Q
}%
    \graphtemp=\baselineskip
    \multiply\graphtemp by -1
    \divide\graphtemp by 2
    \advance\graphtemp by .5ex
    \advance\graphtemp by 0.349in
    \rlap{\kern 2.093in\lower\graphtemp\hbox to 0pt{\hss $r(c_i,\$)$\hss}}%
\pdfliteral{
q [] 0 d 1 J 1 j
0.576 w
284.616 -25.128 m
284.616 -29.740682 280.876682 -33.48 276.264 -33.48 c
271.651318 -33.48 267.912 -29.740682 267.912 -25.128 c
267.912 -20.515318 271.651318 -16.776 276.264 -16.776 c
280.876682 -16.776 284.616 -20.515318 284.616 -25.128 c
S
0.072 w
q 0 g
260.712 -23.328 m
267.912 -25.128 l
260.712 -26.928 l
260.712 -23.328 l
B Q
0.576 w
200.952 -25.128 m
260.712 -25.128 l
S
Q
}%
    \graphtemp=\baselineskip
    \multiply\graphtemp by -1
    \divide\graphtemp by 2
    \advance\graphtemp by .5ex
    \advance\graphtemp by 0.349in
    \rlap{\kern 3.256in\lower\graphtemp\hbox to 0pt{\hss $s(c_{i+1},P(a_i))\;\,$\hss}}%
    \graphtemp=\baselineskip
    \multiply\graphtemp by 1
    \divide\graphtemp by 2
    \advance\graphtemp by .5ex
    \advance\graphtemp by 0.349in
    \rlap{\kern 3.256in\lower\graphtemp\hbox to 0pt{\hss $u:=\now$\hss}}%
\pdfliteral{
q [] 0 d 1 J 1 j
0.576 w
326.52 -83.736 m
326.52 -88.348682 322.780682 -92.088 318.168 -92.088 c
313.555318 -92.088 309.816 -88.348682 309.816 -83.736 c
309.816 -79.123318 313.555318 -75.384 318.168 -75.384 c
322.780682 -75.384 326.52 -79.123318 326.52 -83.736 c
h q 0.5 g
B Q
0.072 w
q 0 g
310.536 -69.984 m
313.272 -76.896 l
307.656 -72.072 l
310.536 -69.984 l
B Q
0.576 w
281.16 -31.896 m
309.096 -71.064 l
S
Q
}%
    \graphtemp=.5ex
    \advance\graphtemp by 0.756in
    \rlap{\kern 4.128in\lower\graphtemp\hbox to 0pt{\hss \makebox[0pt][l]{$\;\now\geq u+a_i$}\hss}}%
\pdfliteral{
q [] 0 d 1 J 1 j
0.576 w
368.352 -25.128 m
368.352 -29.740682 364.612682 -33.48 360 -33.48 c
355.387318 -33.48 351.648 -29.740682 351.648 -25.128 c
351.648 -20.515318 355.387318 -16.776 360 -16.776 c
364.612682 -16.776 368.352 -20.515318 368.352 -25.128 c
h q 0.5 g
B Q
0.072 w
q 0 g
344.448 -23.328 m
351.648 -25.128 l
344.448 -26.928 l
344.448 -23.328 l
B Q
0.576 w
284.616 -25.128 m
344.448 -25.128 l
S
Q
}%
    \graphtemp=\baselineskip
    \multiply\graphtemp by -1
    \divide\graphtemp by 2
    \advance\graphtemp by .5ex
    \advance\graphtemp by 0.349in
    \rlap{\kern 4.419in\lower\graphtemp\hbox to 0pt{\hss $r(c_{i+1},\chi)$\hss}}%
\pdfliteral{
q [] 0 d 1 J 1 j
0.576 w
452.088 -25.128 m
452.088 -29.740682 448.348682 -33.48 443.736 -33.48 c
439.123318 -33.48 435.384 -29.740682 435.384 -25.128 c
435.384 -20.515318 439.123318 -16.776 443.736 -16.776 c
448.348682 -16.776 452.088 -20.515318 452.088 -25.128 c
S
0.072 w
q 0 g
428.184 -23.328 m
435.384 -25.128 l
428.184 -26.928 l
428.184 -23.328 l
B Q
0.576 w
368.352 -25.128 m
428.184 -25.128 l
S
Q
}%
    \graphtemp=\baselineskip
    \multiply\graphtemp by -1
    \divide\graphtemp by 2
    \advance\graphtemp by .5ex
    \advance\graphtemp by 0.349in
    \rlap{\kern 5.581in\lower\graphtemp\hbox to 0pt{\hss $s(c_{i},\chi)$\hss}}%
    \graphtemp=\baselineskip
    \multiply\graphtemp by 1
    \divide\graphtemp by 2
    \advance\graphtemp by .5ex
    \advance\graphtemp by 0.349in
    \rlap{\kern 5.581in\lower\graphtemp\hbox to 0pt{\hss $s(c_{i+1},\$)$\hss}}%
\pdfliteral{
q [] 0 d 1 J 1 j
0.576 w
450.432 -25.128 m
450.432 -28.826099 447.434099 -31.824 443.736 -31.824 c
440.037901 -31.824 437.04 -28.826099 437.04 -25.128 c
437.04 -21.429901 440.037901 -18.432 443.736 -18.432 c
447.434099 -18.432 450.432 -21.429901 450.432 -25.128 c
S
410.184 -83.736 m
410.184 -88.348682 406.444682 -92.088 401.832 -92.088 c
397.219318 -92.088 393.48 -88.348682 393.48 -83.736 c
393.48 -79.123318 397.219318 -75.384 401.832 -75.384 c
406.444682 -75.384 410.184 -79.123318 410.184 -83.736 c
S
0.072 w
q 0 g
386.28 -81.936 m
393.48 -83.736 l
386.28 -85.536 l
386.28 -81.936 l
B Q
0.576 w
326.52 -83.736 m
386.28 -83.736 l
S
Q
}%
    \graphtemp=\baselineskip
    \multiply\graphtemp by -1
    \divide\graphtemp by 2
    \advance\graphtemp by .5ex
    \advance\graphtemp by 1.163in
    \rlap{\kern 5.000in\lower\graphtemp\hbox to 0pt{\hss $s(c_{i},\$)$\hss}}%
\pdfliteral{
q [] 0 d 1 J 1 j
0.576 w
408.528 -83.736 m
408.528 -87.434099 405.530099 -90.432 401.832 -90.432 c
398.133901 -90.432 395.136 -87.434099 395.136 -83.736 c
395.136 -80.037901 398.133901 -77.04 401.832 -77.04 c
405.530099 -77.04 408.528 -80.037901 408.528 -83.736 c
S
Q
}%
    \graphtemp=.5ex
    \advance\graphtemp by 1.860in
    \rlap{\kern 0.930in\lower\graphtemp\hbox to 0pt{\hss $c_i$: Customer $i$ ($i=1,\dots,n{-}1$); {\it Chloe$_i$}\hss}}%
\pdfliteral{
q [] 0 d 1 J 1 j
0.576 w
33.48 -159.048 m
33.48 -163.660682 29.740682 -167.4 25.128 -167.4 c
20.515318 -167.4 16.776 -163.660682 16.776 -159.048 c
16.776 -154.435318 20.515318 -150.696 25.128 -150.696 c
29.740682 -150.696 33.48 -154.435318 33.48 -159.048 c
S
0.072 w
q 0 g
9.576 -157.248 m
16.776 -159.048 l
9.576 -160.848 l
9.576 -157.248 l
B Q
0.576 w
0 -159.048 m
9.576 -159.048 l
S
159.048 -217.656 m
159.048 -222.268682 155.308682 -226.008 150.696 -226.008 c
146.083318 -226.008 142.344 -222.268682 142.344 -217.656 c
142.344 -213.043318 146.083318 -209.304 150.696 -209.304 c
155.308682 -209.304 159.048 -213.043318 159.048 -217.656 c
h q 0.5 g
B Q
242.784 -217.656 m
242.784 -222.268682 239.044682 -226.008 234.432 -226.008 c
229.819318 -226.008 226.08 -222.268682 226.08 -217.656 c
226.08 -213.043318 229.819318 -209.304 234.432 -209.304 c
239.044682 -209.304 242.784 -213.043318 242.784 -217.656 c
S
0.072 w
q 0 g
218.88 -215.856 m
226.08 -217.656 l
218.88 -219.456 l
218.88 -215.856 l
B Q
0.576 w
159.048 -217.656 m
218.88 -217.656 l
S
Q
}%
    \graphtemp=\baselineskip
    \multiply\graphtemp by -1
    \divide\graphtemp by 2
    \advance\graphtemp by .5ex
    \advance\graphtemp by 3.023in
    \rlap{\kern 2.674in\lower\graphtemp\hbox to 0pt{\hss $s(e_{i},\$)$\hss}}%
\pdfliteral{
q [] 0 d 1 J 1 j
0.576 w
326.52 -217.656 m
326.52 -222.268682 322.780682 -226.008 318.168 -226.008 c
313.555318 -226.008 309.816 -222.268682 309.816 -217.656 c
309.816 -213.043318 313.555318 -209.304 318.168 -209.304 c
322.780682 -209.304 326.52 -213.043318 326.52 -217.656 c
S
0.072 w
q 0 g
302.544 -215.856 m
309.744 -217.656 l
302.544 -219.456 l
302.544 -215.856 l
B Q
0.576 w
242.784 -217.656 m
302.544 -217.656 l
S
Q
}%
    \graphtemp=\baselineskip
    \multiply\graphtemp by -1
    \divide\graphtemp by 2
    \advance\graphtemp by .5ex
    \advance\graphtemp by 3.023in
    \rlap{\kern 3.837in\lower\graphtemp\hbox to 0pt{\hss $r(e_{i},\$)$\hss}}%
\pdfliteral{
q [] 0 d 1 J 1 j
0.576 w
324.864 -217.656 m
324.864 -221.354099 321.866099 -224.352 318.168 -224.352 c
314.469901 -224.352 311.472 -221.354099 311.472 -217.656 c
311.472 -213.957901 314.469901 -210.96 318.168 -210.96 c
321.866099 -210.96 324.864 -213.957901 324.864 -217.656 c
S
284.616 -159.048 m
284.616 -163.660682 280.876682 -167.4 276.264 -167.4 c
271.651318 -167.4 267.912 -163.660682 267.912 -159.048 c
267.912 -154.435318 271.651318 -150.696 276.264 -150.696 c
280.876682 -150.696 284.616 -154.435318 284.616 -159.048 c
h q 0.5 g
B Q
0.072 w
q 0 g
265.752 -170.712 m
271.44 -165.888 l
268.704 -172.8 l
265.752 -170.712 l
B Q
0.576 w
239.256 -210.888 m
267.192 -171.72 l
S
Q
}%
    \graphtemp=.5ex
    \advance\graphtemp by 2.616in
    \rlap{\kern 3.547in\lower\graphtemp\hbox to 0pt{\hss \makebox[0pt][l]{$\;\;r(e_{i},\chi)$}\hss}}%
\pdfliteral{
q [] 0 d 1 J 1 j
0.576 w
368.352 -159.048 m
368.352 -163.660682 364.612682 -167.4 360 -167.4 c
355.387318 -167.4 351.648 -163.660682 351.648 -159.048 c
351.648 -154.435318 355.387318 -150.696 360 -150.696 c
364.612682 -150.696 368.352 -154.435318 368.352 -159.048 c
S
0.072 w
q 0 g
344.448 -157.248 m
351.648 -159.048 l
344.448 -160.848 l
344.448 -157.248 l
B Q
0.576 w
284.616 -159.048 m
344.448 -159.048 l
S
Q
}%
    \graphtemp=\baselineskip
    \multiply\graphtemp by -1
    \divide\graphtemp by 2
    \advance\graphtemp by .5ex
    \advance\graphtemp by 2.209in
    \rlap{\kern 4.419in\lower\graphtemp\hbox to 0pt{\hss $s(e_{i-1},\chi)$\hss}}%
\pdfliteral{
q [] 0 d 1 J 1 j
0.576 w
452.088 -159.048 m
452.088 -163.660682 448.348682 -167.4 443.736 -167.4 c
439.123318 -167.4 435.384 -163.660682 435.384 -159.048 c
435.384 -154.435318 439.123318 -150.696 443.736 -150.696 c
448.348682 -150.696 452.088 -154.435318 452.088 -159.048 c
S
0.072 w
q 0 g
428.184 -157.248 m
435.384 -159.048 l
428.184 -160.848 l
428.184 -157.248 l
B Q
0.576 w
368.352 -159.048 m
428.184 -159.048 l
S
Q
}%
    \graphtemp=\baselineskip
    \multiply\graphtemp by -1
    \divide\graphtemp by 2
    \advance\graphtemp by .5ex
    \advance\graphtemp by 2.209in
    \rlap{\kern 5.581in\lower\graphtemp\hbox to 0pt{\hss $r(e_{i-1},\$)$\hss}}%
\pdfliteral{
q [] 0 d 1 J 1 j
0.576 w
450.432 -159.048 m
450.432 -162.746099 447.434099 -165.744 443.736 -165.744 c
440.037901 -165.744 437.04 -162.746099 437.04 -159.048 c
437.04 -155.349901 440.037901 -152.352 443.736 -152.352 c
447.434099 -152.352 450.432 -155.349901 450.432 -159.048 c
S
75.312 -217.656 m
75.312 -222.268682 71.572682 -226.008 66.96 -226.008 c
62.347318 -226.008 58.608 -222.268682 58.608 -217.656 c
58.608 -213.043318 62.347318 -209.304 66.96 -209.304 c
71.572682 -209.304 75.312 -213.043318 75.312 -217.656 c
S
0.072 w
q 0 g
135.144 -215.856 m
142.344 -217.656 l
135.144 -219.456 l
135.144 -215.856 l
B Q
0.576 w
75.384 -217.656 m
135.144 -217.656 l
S
Q
}%
    \graphtemp=\baselineskip
    \multiply\graphtemp by -1
    \divide\graphtemp by 2
    \advance\graphtemp by .5ex
    \advance\graphtemp by 3.023in
    \rlap{\kern 1.512in\lower\graphtemp\hbox to 0pt{\hss $r(e_{i-1},P(a_{i-1}))$\hss}}%
\pdfliteral{
q [] 0 d 1 J 1 j
0.576 w
0.072 w
q 0 g
59.4 -203.976 m
62.136 -210.888 l
56.448 -206.064 l
59.4 -203.976 l
B Q
0.576 w
29.952 -165.888 m
57.96 -204.984 l
S
Q
}%
    \graphtemp=.5ex
    \advance\graphtemp by 2.616in
    \rlap{\kern 0.640in\lower\graphtemp\hbox to 0pt{\hss \makebox[0pt][l]{$\;\;r(e_{i},G(d_i))$}\hss}}%
    \graphtemp=.5ex
    \advance\graphtemp by 4.651in
    \rlap{\kern 0.698in\lower\graphtemp\hbox to 0pt{\hss $c_0$: Customer $0$; {\it Alice}\hss}}%
\pdfliteral{
q [] 0 d 1 J 1 j
0.576 w
75.312 -360 m
75.312 -364.612682 71.572682 -368.352 66.96 -368.352 c
62.347318 -368.352 58.608 -364.612682 58.608 -360 c
58.608 -355.387318 62.347318 -351.648 66.96 -351.648 c
71.572682 -351.648 75.312 -355.387318 75.312 -360 c
S
0.072 w
q 0 g
51.408 -358.2 m
58.608 -360 l
51.408 -361.8 l
51.408 -358.2 l
B Q
0.576 w
41.832 -360 m
51.408 -360 l
S
159.048 -360 m
159.048 -364.612682 155.308682 -368.352 150.696 -368.352 c
146.083318 -368.352 142.344 -364.612682 142.344 -360 c
142.344 -355.387318 146.083318 -351.648 150.696 -351.648 c
155.308682 -351.648 159.048 -355.387318 159.048 -360 c
h q 0 g
B Q
0.072 w
q 0 g
135.144 -358.2 m
142.344 -360 l
135.144 -361.8 l
135.144 -358.2 l
B Q
0.576 w
75.384 -360 m
135.144 -360 l
S
Q
}%
    \graphtemp=\baselineskip
    \multiply\graphtemp by -1
    \divide\graphtemp by 2
    \advance\graphtemp by .5ex
    \advance\graphtemp by 5.000in
    \rlap{\kern 1.512in\lower\graphtemp\hbox to 0pt{\hss $r(e_0,G(d_0))\;$\hss}}%
\pdfliteral{
q [] 0 d 1 J 1 j
0.576 w
242.784 -360 m
242.784 -364.612682 239.044682 -368.352 234.432 -368.352 c
229.819318 -368.352 226.08 -364.612682 226.08 -360 c
226.08 -355.387318 229.819318 -351.648 234.432 -351.648 c
239.044682 -351.648 242.784 -355.387318 242.784 -360 c
S
0.072 w
q 0 g
218.88 -358.2 m
226.08 -360 l
218.88 -361.8 l
218.88 -358.2 l
B Q
0.576 w
159.048 -360 m
218.88 -360 l
S
Q
}%
    \graphtemp=\baselineskip
    \multiply\graphtemp by -1
    \divide\graphtemp by 2
    \advance\graphtemp by .5ex
    \advance\graphtemp by 5.000in
    \rlap{\kern 2.674in\lower\graphtemp\hbox to 0pt{\hss $s(e_0,\$)$\hss}}%
\pdfliteral{
q [] 0 d 1 J 1 j
0.576 w
326.52 -360 m
326.52 -364.612682 322.780682 -368.352 318.168 -368.352 c
313.555318 -368.352 309.816 -364.612682 309.816 -360 c
309.816 -355.387318 313.555318 -351.648 318.168 -351.648 c
322.780682 -351.648 326.52 -355.387318 326.52 -360 c
S
0.072 w
q 0 g
302.544 -358.2 m
309.744 -360 l
302.544 -361.8 l
302.544 -358.2 l
B Q
0.576 w
242.784 -360 m
302.544 -360 l
S
Q
}%
    \graphtemp=\baselineskip
    \multiply\graphtemp by -1
    \divide\graphtemp by 2
    \advance\graphtemp by .5ex
    \advance\graphtemp by 5.000in
    \rlap{\kern 3.837in\lower\graphtemp\hbox to 0pt{\hss $r(e_0,\$)$\hss}}%
\pdfliteral{
q [] 0 d 1 J 1 j
0.576 w
324.864 -360 m
324.864 -363.698099 321.866099 -366.696 318.168 -366.696 c
314.469901 -366.696 311.472 -363.698099 311.472 -360 c
311.472 -356.301901 314.469901 -353.304 318.168 -353.304 c
321.866099 -353.304 324.864 -356.301901 324.864 -360 c
S
284.616 -301.392 m
284.616 -306.004682 280.876682 -309.744 276.264 -309.744 c
271.651318 -309.744 267.912 -306.004682 267.912 -301.392 c
267.912 -296.779318 271.651318 -293.04 276.264 -293.04 c
280.876682 -293.04 284.616 -296.779318 284.616 -301.392 c
S
0.072 w
q 0 g
265.752 -313.056 m
271.44 -308.232 l
268.704 -315.144 l
265.752 -313.056 l
B Q
0.576 w
239.256 -353.16 m
267.192 -314.064 l
S
Q
}%
    \graphtemp=.5ex
    \advance\graphtemp by 4.593in
    \rlap{\kern 3.547in\lower\graphtemp\hbox to 0pt{\hss \makebox[0pt][l]{$\;\;r(e_0,\chi)$}\hss}}%
\pdfliteral{
q [] 0 d 1 J 1 j
0.576 w
282.96 -301.392 m
282.96 -305.090099 279.962099 -308.088 276.264 -308.088 c
272.565901 -308.088 269.568 -305.090099 269.568 -301.392 c
269.568 -297.693901 272.565901 -294.696 276.264 -294.696 c
279.962099 -294.696 282.96 -297.693901 282.96 -301.392 c
S
Q
}%
    \graphtemp=.5ex
    \advance\graphtemp by 5.814in
    \rlap{\kern 0.698in\lower\graphtemp\hbox to 0pt{\hss $c_n$: Customer $n$; {\it Bob}\hss}}%
\pdfliteral{
q [] 0 d 1 J 1 j
0.576 w
75.312 -443.736 m
75.312 -448.348682 71.572682 -452.088 66.96 -452.088 c
62.347318 -452.088 58.608 -448.348682 58.608 -443.736 c
58.608 -439.123318 62.347318 -435.384 66.96 -435.384 c
71.572682 -435.384 75.312 -439.123318 75.312 -443.736 c
S
0.072 w
q 0 g
51.408 -441.936 m
58.608 -443.736 l
51.408 -445.536 l
51.408 -441.936 l
B Q
0.576 w
41.832 -443.736 m
51.408 -443.736 l
S
159.048 -443.736 m
159.048 -448.348682 155.308682 -452.088 150.696 -452.088 c
146.083318 -452.088 142.344 -448.348682 142.344 -443.736 c
142.344 -439.123318 146.083318 -435.384 150.696 -435.384 c
155.308682 -435.384 159.048 -439.123318 159.048 -443.736 c
h q 0.5 g
B Q
0.072 w
q 0 g
135.144 -441.936 m
142.344 -443.736 l
135.144 -445.536 l
135.144 -441.936 l
B Q
0.576 w
75.384 -443.736 m
135.144 -443.736 l
S
Q
}%
    \graphtemp=\baselineskip
    \multiply\graphtemp by -1
    \divide\graphtemp by 2
    \advance\graphtemp by .5ex
    \advance\graphtemp by 6.163in
    \rlap{\kern 1.512in\lower\graphtemp\hbox to 0pt{\hss $r(e_{n-1},P(a_{n-1}))$\hss}}%
\pdfliteral{
q [] 0 d 1 J 1 j
0.576 w
242.784 -443.736 m
242.784 -448.348682 239.044682 -452.088 234.432 -452.088 c
229.819318 -452.088 226.08 -448.348682 226.08 -443.736 c
226.08 -439.123318 229.819318 -435.384 234.432 -435.384 c
239.044682 -435.384 242.784 -439.123318 242.784 -443.736 c
S
0.072 w
q 0 g
218.88 -441.936 m
226.08 -443.736 l
218.88 -445.536 l
218.88 -441.936 l
B Q
0.576 w
159.048 -443.736 m
218.88 -443.736 l
S
Q
}%
    \graphtemp=\baselineskip
    \multiply\graphtemp by -1
    \divide\graphtemp by 2
    \advance\graphtemp by .5ex
    \advance\graphtemp by 6.163in
    \rlap{\kern 2.674in\lower\graphtemp\hbox to 0pt{\hss $s(e_{n-1},\chi)$\hss}}%
\pdfliteral{
q [] 0 d 1 J 1 j
0.576 w
326.52 -443.736 m
326.52 -448.348682 322.780682 -452.088 318.168 -452.088 c
313.555318 -452.088 309.816 -448.348682 309.816 -443.736 c
309.816 -439.123318 313.555318 -435.384 318.168 -435.384 c
322.780682 -435.384 326.52 -439.123318 326.52 -443.736 c
S
0.072 w
q 0 g
302.544 -441.936 m
309.744 -443.736 l
302.544 -445.536 l
302.544 -441.936 l
B Q
0.576 w
242.784 -443.736 m
302.544 -443.736 l
S
Q
}%
    \graphtemp=\baselineskip
    \multiply\graphtemp by -1
    \divide\graphtemp by 2
    \advance\graphtemp by .5ex
    \advance\graphtemp by 6.163in
    \rlap{\kern 3.837in\lower\graphtemp\hbox to 0pt{\hss $r(e_{n-1},\$)$\hss}}%
\pdfliteral{
q [] 0 d 1 J 1 j
0.576 w
324.864 -443.736 m
324.864 -447.434099 321.866099 -450.432 318.168 -450.432 c
314.469901 -450.432 311.472 -447.434099 311.472 -443.736 c
311.472 -440.037901 314.469901 -437.04 318.168 -437.04 c
321.866099 -437.04 324.864 -440.037901 324.864 -443.736 c
S
Q
}%
    \hbox{\vrule depth6.512in width0pt height 0pt}%
    \kern 6.279in
  }%
}%

%% file: Chloe.tex
\expandafter\ifx\csname graph\endcsname\relax
   \csname newbox\expandafter\endcsname\csname graph\endcsname
\fi
\ifx\graphtemp\undefined
  \csname newdimen\endcsname\graphtemp
\fi
\expandafter\setbox\csname graph\endcsname
 =\vtop{\vskip 0pt\hbox{%
    \graphtemp=.5ex
    \advance\graphtemp by 0.000in
    \rlap{\kern 0.100in\lower\graphtemp\hbox to 0pt{\hss $e_{i-1}$\hss}}%
    \graphtemp=.5ex
    \advance\graphtemp by 2.400in
    \rlap{\kern 0.100in\lower\graphtemp\hbox to 0pt{\hss Escrow $i{-}1$\hss}}%
\pdfliteral{
q [] 0 d 1 J 1 j
0.576 w
0.576 w
7.2 -7.2 m
7.2 -164.88 l
S
Q
}%
    \graphtemp=.5ex
    \advance\graphtemp by 0.000in
    \rlap{\kern 1.300in\lower\graphtemp\hbox to 0pt{\hss $c_i$\hss}}%
    \graphtemp=.5ex
    \advance\graphtemp by 2.400in
    \rlap{\kern 1.300in\lower\graphtemp\hbox to 0pt{\hss Chloe$_i$\hss}}%
\pdfliteral{
q [] 0 d 1 J 1 j
0.576 w
93.6 -7.92 m
93.6 -164.88 l
S
Q
}%
    \graphtemp=.5ex
    \advance\graphtemp by 0.000in
    \rlap{\kern 2.500in\lower\graphtemp\hbox to 0pt{\hss $e_{i}$\hss}}%
    \graphtemp=.5ex
    \advance\graphtemp by 2.400in
    \rlap{\kern 2.500in\lower\graphtemp\hbox to 0pt{\hss Escrow $i$\hss}}%
\pdfliteral{
q [] 0 d 1 J 1 j
0.576 w
180 -7.92 m
180 -164.88 l
S
0.072 w
q 0 g
86.76 -25.848 m
93.6 -28.8 l
86.184 -29.376 l
86.76 -25.848 l
B Q
0.576 w
7.2 -14.4 m
86.472 -27.648 l
S
Q
}%
    \graphtemp=\baselineskip
    \multiply\graphtemp by -1
    \divide\graphtemp by 2
    \advance\graphtemp by .5ex
    \advance\graphtemp by 0.300in
    \rlap{\kern 0.700in\lower\graphtemp\hbox to 0pt{\hss $P(a_{i-1})$\hss}}%
\pdfliteral{
q [] 0 d 1 J 1 j
0.576 w
0.072 w
q 0 g
173.16 -54.648 m
180 -57.6 l
172.584 -58.176 l
173.16 -54.648 l
B Q
0.576 w
93.6 -43.2 m
172.872 -56.448 l
S
Q
}%
    \graphtemp=\baselineskip
    \multiply\graphtemp by -1
    \divide\graphtemp by 2
    \advance\graphtemp by .5ex
    \advance\graphtemp by 0.700in
    \rlap{\kern 1.900in\lower\graphtemp\hbox to 0pt{\hss $\$$\hss}}%
\pdfliteral{
q [] 0 d 1 J 1 j
0.576 w
0.072 w
q 0 g
101.016 -101.376 m
93.6 -100.8 l
100.44 -97.848 l
101.016 -101.376 l
B Q
0.576 w
180 -86.4 m
100.728 -99.648 l
S
Q
}%
    \graphtemp=\baselineskip
    \multiply\graphtemp by -1
    \divide\graphtemp by 2
    \advance\graphtemp by .5ex
    \advance\graphtemp by 1.300in
    \rlap{\kern 1.900in\lower\graphtemp\hbox to 0pt{\hss $\$$\hss}}%
    \graphtemp=\baselineskip
    \multiply\graphtemp by 1
    \divide\graphtemp by 2
    \advance\graphtemp by .5ex
    \advance\graphtemp by 1.300in
    \rlap{\kern 1.900in\lower\graphtemp\hbox to 0pt{\hss or $\chi$\hss}}%
\pdfliteral{
q [] 0 d 1 J 1 j
0.576 w
0.072 w
q 0 g
14.616 -130.176 m
7.2 -129.6 l
14.04 -126.648 l
14.616 -130.176 l
B Q
0.576 w
93.6 -115.2 m
14.328 -128.448 l
S
Q
}%
    \graphtemp=\baselineskip
    \multiply\graphtemp by -1
    \divide\graphtemp by 2
    \advance\graphtemp by .5ex
    \advance\graphtemp by 1.700in
    \rlap{\kern 0.700in\lower\graphtemp\hbox to 0pt{\hss $\chi$\hss}}%
\pdfliteral{
q [] 0 d 1 J 1 j
0.576 w
0.072 w
q 0 g
86.76 -155.448 m
93.6 -158.4 l
86.184 -158.976 l
86.76 -155.448 l
B Q
0.576 w
7.2 -144 m
86.472 -157.248 l
S
Q
}%
    \graphtemp=\baselineskip
    \multiply\graphtemp by -1
    \divide\graphtemp by 2
    \advance\graphtemp by .5ex
    \advance\graphtemp by 2.100in
    \rlap{\kern 0.700in\lower\graphtemp\hbox to 0pt{\hss $\$$\hss}}%
    \graphtemp=.5ex
    \advance\graphtemp by 0.500in
    \rlap{\kern 1.400in\lower\graphtemp\hbox to 0pt{\hss $\}$\hss}}%
    \graphtemp=.5ex
    \advance\graphtemp by 0.500in
    \rlap{\kern 1.550in\lower\graphtemp\hbox to 0pt{\hss $< \epsilon$\hss}}%
    \graphtemp=.5ex
    \advance\graphtemp by 1.530in
    \rlap{\kern 1.400in\lower\graphtemp\hbox to 0pt{\hss $\}$\hss}}%
    \graphtemp=.5ex
    \advance\graphtemp by 1.530in
    \rlap{\kern 1.550in\lower\graphtemp\hbox to 0pt{\hss $< \epsilon$\hss}}%
    \graphtemp=.5ex
    \advance\graphtemp by 1.930in
    \rlap{\kern 0.200in\lower\graphtemp\hbox to 0pt{\hss $\}$\hss}}%
    \graphtemp=.5ex
    \advance\graphtemp by 1.930in
    \rlap{\kern 0.350in\lower\graphtemp\hbox to 0pt{\hss $< \epsilon$\hss}}%
    \graphtemp=.5ex
    \advance\graphtemp by 1.000in
    \rlap{\kern 2.600in\lower\graphtemp\hbox to 0pt{\hss $\left.\rule{0pt}{14pt}\right\}$\hss}}%
    \graphtemp=.5ex
    \advance\graphtemp by 1.000in
    \rlap{\kern 2.800in\lower\graphtemp\hbox to 0pt{\hss $\leq d_i$\hss}}%
    \graphtemp=.5ex
    \advance\graphtemp by 1.000in
    \rlap{\kern 0.250in\lower\graphtemp\hbox to 0pt{\hss $\left.\rule{0pt}{56pt}\right\}$\hss}}%
    \graphtemp=.5ex
    \advance\graphtemp by 1.000in
    \rlap{\kern 1.130in\lower\graphtemp\hbox to 0pt{\hss $\leq 2\cdot\Phi\cdot\epsilon+\Phi\cdot d_i+4\cdot \Delta$\hss}}%
    \graphtemp=.5ex
    \advance\graphtemp by 0.200in
    \rlap{\kern 0.000in\lower\graphtemp\hbox to 0pt{\hss $u$\hss}}%
    \graphtemp=.5ex
    \advance\graphtemp by 0.350in
    \rlap{\kern 1.400in\lower\graphtemp\hbox to 0pt{\hss $t$\hss}}%
    \graphtemp=.5ex
    \advance\graphtemp by 0.600in
    \rlap{\kern 1.200in\lower\graphtemp\hbox to 0pt{\hss $t_0$\hss}}%
    \graphtemp=.5ex
    \advance\graphtemp by 0.800in
    \rlap{\kern 2.700in\lower\graphtemp\hbox to 0pt{\hss $w$\hss}}%
    \graphtemp=.5ex
    \advance\graphtemp by 1.800in
    \rlap{\kern 0.000in\lower\graphtemp\hbox to 0pt{\hss $v$\hss}}%
    \hbox{\vrule depth2.400in width0pt height 0pt}%
    \kern 2.800in
  }%
}%